\pdfoutput=1
\documentclass{article}
\usepackage{arxiv}

\usepackage{graphicx}
\usepackage[utf8]{inputenc}
\usepackage[T1]{fontenc}
\usepackage{url}
\usepackage{amsfonts}
\usepackage{microtype}
\usepackage{mathptmx}
\usepackage{amsmath}
\usepackage{amsthm}
\usepackage{amssymb}
\usepackage{enumerate}
\usepackage{bussproofs}
\usepackage{booktabs}
\usepackage{latexsym}
\usepackage[hidelinks]{hyperref}
\hypersetup{pdftitle={A Lambda Calculus Based on Quantum Logic}, pdfauthor={Kenji Tokuo}}

\newtheoremstyle{thm-break}
  {10pt}
  {3pt}
  {\normalfont}
  {}
  {\bfseries}
  {}
  {\newline}
  {\thmname{#1}\thmnumber{ #2}\thmnote{ (#3)}\par\vspace{0.5\baselineskip}}

\theoremstyle{thm-break}

\newtheorem{theorem}{Theorem}[section]
\newtheorem{lemma}[theorem]{Lemma}
\newtheorem{proposition}[theorem]{Proposition}
\newtheorem{corollary}[theorem]{Corollary}
\newtheorem{definition}[theorem]{Definition}
\newtheorem{example}[theorem]{Example}
\newtheorem{remark}[theorem]{Remark}

\newcommand{\FV}{\mathrm{FV}}
\newcommand{\QuLam}{\textnormal{QuLam}}
\providecommand{\llbracket}{[\![}
\providecommand{\rrbracket}{]\!]}

\title{A Lambda Calculus Based on Quantum Logic}
\author{Kenji Tokuo\\
Department of Information Engineering, Oita College\\
National Institute of Technology, Oita 870-0152, Japan\\
\texttt{tokuo@oita-ct.ac.jp}}
\date{}

\begin{document}

\maketitle

\begin{abstract}
We propose a typed lambda calculus based on quantum logic. We first present a natural deduction system for quantum logic with the Sasaki hook as its implication operation and then formulate a typed lambda calculus connected to this system through the Curry–Howard correspondence. We prove soundness and completeness of the former with respect to orthomodular lattices and strong normalization of the latter. The resulting lambda calculus has the characteristic feature that only one assumption is available in the typing derivation that serves as the premise of $\lambda$-abstraction. Under this restriction, ordinary currying is not available in general, but it is still possible to perform $\lambda$-abstraction by combining multiple assumptions into a single argument of product type and to construct a term corresponding to a deduction using the Sasaki projection.
\end{abstract}


\section{Introduction}\label{sec:intro}

\subsection{Background}
Quantum logic originates with Birkhoff and von Neumann~\cite{Birk1936}. Orthomodular lattices are known as an algebraic semantics for quantum logic. Since implication in quantum logic cannot in general satisfy a satisfactory set of properties, the Sasaki hook is often used as a second-best candidate for implication.

Nishimura developed sequent systems for quantum logic~\cite{Nish1980,Nish2009}. Based on these systems, a natural deduction system for quantum logic and an associated typed lambda calculus were introduced in~\cite{Toku2022}. Another natural deduction system for quantum logic, SOM~\cite{Korn2024}, uses ordered sequences of assumptions, and the deduction theorem for the Sasaki hook holds for these ordered sequences.

Meanwhile, lambda calculi for quantum computation have also been studied. Van Tonder proposed a quantum lambda calculus closely related to linear lambda calculus as a computational model for describing quantum algorithms~\cite{Vant2004}. Selinger and Valiron formulated a lambda calculus that handles quantum data under classical control and provided a type system based on affine intuitionistic linear logic~\cite{Seli2005}. Altenkirch and Grattage also introduced the functional quantum programming language QML based on linear logic~\cite{Alte2005}.

\subsection{Contributions}
We revisit the natural deduction system NQ for quantum logic and the typed lambda calculus QuLam, which we introduced in our previous paper~\cite{Toku2022}, and add the missing rules explicitly. We then establish once again the correspondence between the former system and Nishimura's sequent calculus GOM for quantum logic~\cite{Nish1980}, together with the resulting soundness and completeness theorem, and prove strong normalization of QuLam.

The types of QuLam correspond to formulas of quantum logic, and its terms correspond to proofs in quantum logic. In QuLam, there is a restriction that the body of a $\lambda$-abstraction may depend only on its bound argument. Although this restriction prevents currying in general, we show that it is still possible to perform $\lambda$-abstraction by combining multiple assumptions into a single product argument and to construct terms corresponding to deductions based on the Sasaki projection. We also treat negative inference, double negation, and orthomodular lattice semantics.

\subsection{Organization}
The paper is organized as follows. Section~\ref{sec:preliminaries} establishes the formal foundations of NQ and QuLam. This includes soundness and completeness of NQ in Section~\ref{sec:NQ-revision} and strong normalization of QuLam in Section~\ref{sec:lambda}. Section~\ref{sec:CH} describes the correspondence between deductions and terms. Section~\ref{sec:semantics} develops the orthomodular lattice semantics of QuLam. Section~\ref{sec:control} examines the structure of terms in QuLam. Section~\ref{sec:conclusion} concludes the paper.

\section{Formal Foundations}
\label{sec:preliminaries}

\subsection{Orthomodular Lattices}
\label{sec:orthomodular-logic}

We briefly review orthomodular lattices, which provide the semantics for quantum logic. For further details, see standard references such as Kalmbach~\cite{Kalm1983}.

Let $(L,\leq)$ be a partially ordered set. We say that $L$ is a lattice if, for any $a,b \in L$, the join $a \vee b$ and the meet $a \wedge b$ exist. A lattice $L$ is said to be bounded if it has a least element $0$ and a greatest element $1$. For a bounded lattice $L$, suppose that a map $\perp : L \to L$ satisfies the following conditions for all $a,b \in L$:
\begin{enumerate}[(i)]
  \item $a \vee a^{\perp} = 1$ and $a \wedge a^{\perp} = 0$.
  \item $(a^{\perp})^{\perp} = a$.
  \item If $a \leq b$, then $b^{\perp} \leq a^{\perp}$.
\end{enumerate}
Then $(L,\leq,\perp)$ is called an \emph{ortholattice}. Here $a^{\perp}$ denotes the orthocomplement of $a$, which is interpreted as the quantum-logical negation of the proposition $a$. An ortholattice $(L,\leq,\perp)$ is called an \emph{orthomodular lattice} if it further satisfies the following orthomodularity condition:
\begin{enumerate}[(i)]
  \setcounter{enumi}{3}
  \item If $a \leq b$, then $b = a \vee (b \wedge a^{\perp})$.
\end{enumerate}

The lattice of closed subspaces of a Hilbert space, ordered by inclusion, is an example of an orthomodular lattice.

In quantum logic, propositions are interpreted as elements of an orthomodular lattice, and the order $a \leq b$ is interpreted as the entailment relation ``$a$ entails $b$.'' For an orthomodular lattice $L$, elements $a,b \in L$ are said to \emph{commute} if the subortholattice of $L$ generated by $a$ and $b$ is Boolean, or equivalently, if $a$ and $b$ belong to a common Boolean subortholattice. On any Boolean subortholattice, the logical connectives behave as in classical logic. In general, orthomodular lattices need not satisfy the distributive law of $\wedge$ over $\vee$.

We introduce two additional operations, which we use below for the interpretation of implication and for the study of terms.

\begin{definition}[Sasaki operations]\label{def:sasaki-operations}
Let $L$ be an orthomodular lattice and let $a,b\in L$. We define the \emph{Sasaki hook} as $a\rightarrow_S b:=a^\perp\vee(a\wedge b)$ and the \emph{Sasaki projection} of $b$ onto $a$ as $\varphi_a(b):=a\wedge(a^\perp\vee b)$.
\end{definition}

The Sasaki hook has been studied in the context of quantum implication~\cite{Coec2004,Hard1979,Hard1981}.

The two operations satisfy the following adjunction.

\begin{lemma}[Sasaki adjunction]\label{lem:sasaki-adjunction}
For all $a,b,c$ in an orthomodular lattice, $\varphi_a(c)\leq b$ if and only if $c\leq a\rightarrow_S b$.
\end{lemma}

\begin{proof}
Suppose $\varphi_a(c)\leq b$. Applying the orthomodularity condition to $a^\perp\leq a^\perp\vee c$, we obtain $a^\perp\vee c=a^\perp\vee\varphi_a(c)$. Since $\varphi_a(c)\leq a\wedge b$, we conclude $c\leq a^\perp\vee\varphi_a(c)\leq a\rightarrow_S b$.

Conversely, suppose $c\leq a\rightarrow_S b$. The projection $\varphi_a$ is monotone, and we obtain $a\wedge(a^\perp\vee(a\wedge b))=a\wedge b$ from the orthomodularity condition. Therefore $\varphi_a(c)\leq\varphi_a(a\rightarrow_S b)=a\wedge b\leq b$.
\end{proof}

We use $\rightarrow_S$ for the operation on lattice elements and $\rightarrow$ for the logical connective of NQ and the type constructor of QuLam.

\subsection{Natural Deduction for Quantum Logic}
\label{sec:NQ}
In this subsection, we redefine the natural deduction system NQ for quantum logic, which we presented in our previous paper~\cite{Toku2022}, with the rules supplemented. The language of NQ consists of countably many propositional variables, the constant $\bot$, conjunction $\wedge$, and primitive implication $\rightarrow$. Negation $\neg\alpha$ abbreviates $\alpha\rightarrow\bot$. Assumptions carry labels $u:\alpha$. Contexts in proof figures are finite sets of labeled assumptions, with each label assigned a unique formula. An occurrence $[u:\alpha]$ is discharged by the rule bearing the label $u$. When $\Gamma$ is a set of formulas, $\Gamma\vdash_{\mathrm{NQ}}\alpha$ means that $\alpha$ can be derived in NQ using only formulas in $\Gamma$ as undischarged assumptions. Each assumption is assigned a distinct label as needed. The set $\Gamma$ may contain unused assumptions.

\begin{definition}[Natural deduction system NQ]\label{def:NQ-full}
Proof figures of NQ are generated by the following rules. Here, renamings of discharged assumption labels are understood in the usual way. The schematic letters $\mathcal D,\mathcal D_1,\mathcal D_2$ denote entire premise derivations. The displayed figures provide the final inference of each rule.

\smallskip\noindent\emph{Assumption} (ASM). An occurrence $u:\alpha$ constitutes a derivation of $\alpha$ with that occurrence as its undischarged assumption.
\begin{prooftree}
\AxiomC{$u:\alpha$}
\end{prooftree}

\smallskip\noindent\emph{Ex falso} (EFQ).
\begin{prooftree}
\AxiomC{$\mathcal D$}
\noLine\UnaryInfC{$\bot$}
\RightLabel{$\mathrm{EFQ}$}
\UnaryInfC{$\alpha$}
\end{prooftree}

\smallskip\noindent\emph{Conjunction introduction and elimination.}
\begin{prooftree}
\AxiomC{$\mathcal D_1$}\noLine\UnaryInfC{$\alpha$}
\AxiomC{$\mathcal D_2$}\noLine\UnaryInfC{$\beta$}
\RightLabel{$\wedge\mathrm I$}\BinaryInfC{$\alpha\wedge\beta$}
\end{prooftree}
\begin{prooftree}
\AxiomC{$\mathcal D$}\noLine\UnaryInfC{$\alpha\wedge\beta$}
\RightLabel{$\wedge\mathrm E_l$}\UnaryInfC{$\alpha$}
\end{prooftree}
\begin{prooftree}
\AxiomC{$\mathcal D$}\noLine\UnaryInfC{$\alpha\wedge\beta$}
\RightLabel{$\wedge\mathrm E_r$}\UnaryInfC{$\beta$}
\end{prooftree}

\smallskip\noindent\emph{Implication introduction and elimination.}
\begin{prooftree}
\AxiomC{$[u:\alpha]$}\noLine\UnaryInfC{$\vdots$}\noLine\UnaryInfC{$\beta$}
\RightLabel{$({\rightarrow}\mathrm I)^u$}\UnaryInfC{$\alpha\rightarrow\beta$}
\end{prooftree}
\begin{prooftree}
\AxiomC{$\mathcal D_1$}\noLine\UnaryInfC{$\alpha\rightarrow\beta$}
\AxiomC{$\mathcal D_2$}\noLine\UnaryInfC{$\alpha$}
\RightLabel{${\rightarrow}\mathrm E$}\BinaryInfC{$\beta$}
\end{prooftree}
The premise derivation above $({\rightarrow}\mathrm I)^u$ must have no undischarged assumptions other than $u:\alpha$. Since the application of ${\rightarrow}\mathrm I$ discharges $u:\alpha$, the derivation of the conclusion $\alpha\rightarrow\beta$ has no undischarged assumptions. This remains true when the derivation occurs inside a larger proof.

\smallskip\noindent\emph{Modus tollens} (MT).
\begin{prooftree}
\AxiomC{$\mathcal D_1$}\noLine\UnaryInfC{$\alpha\rightarrow\beta$}
\AxiomC{$\mathcal D_2$}\noLine\UnaryInfC{$\neg\beta$}
\RightLabel{$\mathrm{MT}$}\BinaryInfC{$\neg\alpha$}
\end{prooftree}
The derivation $\mathcal D_1$ must have no undischarged assumptions. The conclusion inherits the undischarged assumptions of $\mathcal D_2$.

\smallskip\noindent\emph{Double-negation introduction and elimination.}
\begin{prooftree}
\AxiomC{$\mathcal D$}\noLine\UnaryInfC{$\alpha$}
\RightLabel{$\neg\neg\mathrm I$}\UnaryInfC{$\neg\neg\alpha$}
\end{prooftree}
\begin{prooftree}
\AxiomC{$\mathcal D$}\noLine\UnaryInfC{$\neg\neg\alpha$}
\RightLabel{$\neg\neg\mathrm E$}\UnaryInfC{$\alpha$}
\end{prooftree}

\smallskip\noindent\emph{Sasaki hook introduction and elimination.}
\begin{prooftree}
\AxiomC{$\mathcal D$}\noLine\UnaryInfC{$\neg(\alpha\wedge\neg(\alpha\wedge\beta))$}
\RightLabel{${\rightarrow}\mathrm I_S$}\UnaryInfC{$\alpha\rightarrow\beta$}
\end{prooftree}
\begin{prooftree}
\AxiomC{$\mathcal D$}\noLine\UnaryInfC{$\alpha\rightarrow\beta$}
\RightLabel{${\rightarrow}\mathrm E_S$}\UnaryInfC{$\neg(\alpha\wedge\neg(\alpha\wedge\beta))$}
\end{prooftree}
The last two rules preserve all undischarged assumptions of their premise derivations. They are distinct from the restricted rule ${\rightarrow}\mathrm I$ and from modus ponens ${\rightarrow}\mathrm E$.
\end{definition}

\begin{remark}[Relation to the earlier rules]\label{rem:NQ-correction}
The rules of NQ above, except ${\rightarrow}\mathrm I_S$ and ${\rightarrow}\mathrm E_S$, already appeared in our previous paper~\cite{Toku2022}. In the proof of Theorem~3.6 of~\cite{Toku2022}, the case corresponding to the O-modular rule of GOM implicitly identified the expression $\neg(\alpha\wedge\neg(\alpha\wedge\beta))$ with the primitive implication $\alpha\rightarrow\beta$. The rule ${\rightarrow}\mathrm I_S$ explicitly formulates the inference used in this part as a rule. The converse rule ${\rightarrow}\mathrm E_S$ is needed to work with formulas containing primitive implication as a subformula. We continue to use the name NQ for the system with these rules added.
\end{remark}

The logical connectives of NQ are interpreted through the operations of orthomodular lattices.

\begin{definition}[Interpretation of formulas]\label{def:formula-interpretation}
Let $L$ be an orthomodular lattice. A \emph{valuation} in $L$ is a map $\nu$ from propositional variables to $L$. We extend it to formulas by $\llbracket p\rrbracket_\nu=\nu(p)$, $\llbracket\bot\rrbracket_\nu=0$, $\llbracket\alpha\wedge\beta\rrbracket_\nu=\llbracket\alpha\rrbracket_\nu\wedge\llbracket\beta\rrbracket_\nu$, and $\llbracket\alpha\rightarrow\beta\rrbracket_\nu=\llbracket\alpha\rrbracket_\nu\rightarrow_S\llbracket\beta\rrbracket_\nu$. For a finite set $\Gamma$ of formulas, $\llbracket\Gamma\rrbracket_\nu$ is the meet of the values of its elements, with $\llbracket\varnothing\rrbracket_\nu=1$. For a finite context of labeled assumptions, we use the same notation for the meet of the values of its underlying formulas, independently of the labels.
\end{definition}

\subsection{Soundness and Completeness}\label{sec:NQ-revision}
Nishimura's GOM is a sequent calculus for quantum logic~\cite{Nish1980}. GOM adopts conjunction and negation as primitive logical connectives and is sound and complete with respect to orthomodular lattice semantics. A sequent of GOM is an expression of the form $\Gamma\vdash\Delta$, where $\Gamma$ and $\Delta$ are finite sets of formulas. Besides the structural rules for sequents and the rules for conjunction and negation, GOM has the following O-modular rule.

\begin{prooftree}
\AxiomC{$\neg\beta\vdash\neg\alpha$}
\AxiomC{$\neg\alpha,\beta\vdash$}
\RightLabel{$\mathrm{O\text{-}modular}$}
\BinaryInfC{$\neg\alpha\vdash\neg\beta$}
\end{prooftree}

In GOM, the Sasaki hook $\alpha\rightarrow\beta$ is used as an abbreviation for $\neg(\alpha\wedge\neg(\alpha\wedge\beta))$. We define the translation $J$ from GOM formulas to NQ formulas as follows:
\begin{itemize}
\item $J(p)=p$.
\item $J(\alpha\wedge\beta)=J(\alpha)\wedge J(\beta)$.
\item $J(\neg\alpha)=J(\alpha)\rightarrow\bot$.
\end{itemize}
Under this translation $J$, a contradiction $\delta\wedge\neg\delta$ representing $\bot$ in GOM is translated as $J(\delta)\wedge\neg J(\delta)$, which is interderivable with $\bot$ in NQ.

Conversely, we fix a propositional variable $q$ and define the translation $H$ from NQ formulas to GOM formulas as follows:
\begin{itemize}
\item $H(p)=p$.
\item $H(\bot)=q\wedge\neg q$.
\item $H(\alpha\wedge\beta)=H(\alpha)\wedge H(\beta)$.
\item $H(\alpha\rightarrow\beta)=\neg(H(\alpha)\wedge\neg(H(\alpha)\wedge H(\beta)))$.
\end{itemize}
Since NQ uses $\neg\alpha$ as an abbreviation for $\alpha\rightarrow\bot$, the image of $\neg\alpha$ under $H$ is generally not literally $\neg H(\alpha)$. These GOM formulas are nevertheless interderivable. By Propositions~3.3 and~3.4 of~\cite{Toku2022}, all contradiction formulas are interderivable, and the Sasaki hook with a contradictory consequent is interderivable with negation. The conjunction and negation rules of GOM, together with cut, permit these equivalences to be used inside larger formulas. Thus the translated formulas are interderivable with the corresponding GOM expressions used in Theorem~3.5 of~\cite{Toku2022}. Moreover, $H(\alpha)$ and $\alpha$ have the same value under every corresponding valuation in every orthomodular lattice, since $(a\wedge(a\wedge b)^\perp)^\perp=a\rightarrow_S b$ and $q\wedge q^\perp=0$.

\begin{lemma}[Admissibility in GOM]\label{lem:NQ-to-GOM}
If $\Gamma\vdash_{\mathrm{NQ}}\alpha$, then $H(\Gamma)\vdash H(\alpha)$ is derivable in GOM.
\end{lemma}
\begin{proof}
We proceed by induction on an NQ derivation of $\alpha$ from $\Gamma$.

Base case. An assumption $\alpha$ is translated into the initial sequent $H(\alpha)\vdash H(\alpha)$ of GOM. By weakening, we obtain $H(\Gamma)\vdash H(\alpha)$.

Induction step. We proceed by cases according to the last rule applied.

\begin{enumerate}[(i)]
\item Case Sasaki hook. For either rule for the Sasaki hook, the premise and conclusion have equal values in every orthomodular lattice. Therefore, by the soundness and completeness of GOM~\cite{Nish1980}, the formulas obtained by translating the premise and conclusion with $H$ are interderivable in GOM. In particular, since the translated conclusion can be derived from the translated premise, we obtain a derivation of the conclusion by applying cut to the derivation of the premise.

\item Case remaining rules. These cases are established in Theorem~3.5 of~\cite{Toku2022}. By the preceding equivalences, we can use the translation $H$ in place of the abbreviated GOM expressions in that proof.\qedhere
\end{enumerate}
\end{proof}

We next turn to the converse translation. The O-modular rule requires ${\rightarrow}\mathrm I_S$.

\begin{lemma}[Translation of normal GOM sequents]\label{lem:GOM-to-NQ}
A GOM sequent $\Gamma\vdash\Delta$ is called \emph{normal} when $\Delta$ contains at most one formula. Suppose a normal sequent $\Gamma\vdash\Delta$ is derivable in GOM by a derivation consisting of normal sequents. If $\Delta=\{\alpha\}$, then $J(\alpha)$ is derivable in NQ from the assumptions $J(\Gamma)$. If $\Delta=\varnothing$, then $\bot$ is derivable in NQ from $J(\Gamma)$.
\end{lemma}

\begin{proof}
We first replace every use of the O-modular rule by the derivation in Proposition~3.2 of~\cite{Toku2022}. This uses the derived rule MP to infer $\Gamma\vdash\beta$ from $\Gamma\vdash\alpha\rightarrow\beta$ and $\Gamma\vdash\alpha$. Note that in GOM the implication in this rule abbreviates $\neg(\alpha\wedge\neg(\alpha\wedge\beta))$. The resulting derivation still consists of normal sequents. For each normal sequent $\Sigma\vdash\Delta$ in this derivation, we then prove that $J(\alpha)$ is derivable in NQ from $J(\Sigma)$ if $\Delta=\{\alpha\}$, and that $\bot$ is derivable in NQ from $J(\Sigma)$ if $\Delta=\varnothing$. We proceed by induction on the number of applications of inference rules in the subderivation ending with each sequent.

Base case. An initial sequent $\alpha\vdash\alpha$ is translated into ASM of NQ.

Induction step. We proceed by cases according to the last rule applied.

\begin{enumerate}[(i)]
\item Case MP. By the induction hypotheses, we obtain NQ derivations from $J(\Sigma)$ of $\neg(J(\alpha)\wedge\neg(J(\alpha)\wedge J(\beta)))$ and $J(\alpha)$. We apply ${\rightarrow}\mathrm I_S$ to the first and ${\rightarrow}\mathrm E$ to the resulting implication and the second derivation. The conclusion is $J(\beta)$.

\item Case remaining rules. We handle the remaining rules for normal sequents using the constructions in Theorem~3.6 of~\cite{Toku2022}. In each case, the translated undischarged assumptions are preserved. An empty succedent corresponds to a derivation of $\bot$. Substitution of derivations for undischarged assumptions, used in cut and the left rules, preserves the restriction on implication introduction in NQ. Moreover, the double-negation left rule can be handled by the present double-negation elimination rule and substitution of derivations.\qedhere
\end{enumerate}

\end{proof}

We next verify that interderivable formulas can be substituted for each other within larger formulas.

\begin{lemma}[Congruence of interderivability]\label{lem:NQ-congruence}
We write $\alpha\simeq\beta$ when $\alpha\vdash_{\mathrm{NQ}}\beta$ and $\beta\vdash_{\mathrm{NQ}}\alpha$. If $\alpha\simeq\alpha'$ and $\beta\simeq\beta'$, then $\neg\alpha\simeq\neg\alpha'$, $\alpha\wedge\beta\simeq\alpha'\wedge\beta'$, and $\alpha\rightarrow\beta\simeq\alpha'\rightarrow\beta'$.
\end{lemma}
\begin{proof}

We obtain the conjunction claim from $\wedge\mathrm I$, $\wedge\mathrm E_l$, and $\wedge\mathrm E_r$. If $\alpha\vdash_{\mathrm{NQ}}\alpha'$, we obtain a closed proof of $\alpha\rightarrow\alpha'$ by the restricted rule ${\rightarrow}\mathrm I$. By applying MT to this proof and an assumption $\neg\alpha'$, we obtain $\neg\alpha'\vdash_{\mathrm{NQ}}\neg\alpha$. By applying the same argument to $\alpha'\vdash_{\mathrm{NQ}}\alpha$, we obtain the negation claim. By the two rules for the Sasaki hook, $\alpha\rightarrow\beta\simeq\neg(\alpha\wedge\neg(\alpha\wedge\beta))$ holds. Therefore, by the congruence properties for conjunction and negation, we obtain $\alpha\rightarrow\beta\simeq\alpha'\rightarrow\beta'$.
\end{proof}

We use Lemmas~\ref{lem:NQ-to-GOM}--\ref{lem:NQ-congruence} to establish soundness and completeness of NQ.

\begin{theorem}[Soundness and completeness of NQ]\label{thm:NQ-corrected-completeness}
For every finite set of assumptions $\Gamma$ and every formula $\alpha$, $\Gamma\vdash_{\mathrm{NQ}}\alpha$ if and only if $\llbracket\Gamma\rrbracket_\nu\leq\llbracket\alpha\rrbracket_\nu$ for every orthomodular lattice $L$ and every valuation $\nu$ in $L$.
\end{theorem}
\begin{proof}
We first prove soundness. Suppose $\Gamma\vdash_{\mathrm{NQ}}\alpha$. By Lemma~\ref{lem:NQ-to-GOM}, $H(\Gamma)\vdash H(\alpha)$ is derivable in GOM. By the soundness of GOM~\cite{Nish1980}, the meet of the values of $H(\Gamma)$ is below the value of $H(\alpha)$. Since $H$ preserves values under the corresponding valuations, we obtain $\llbracket\Gamma\rrbracket_\nu\leq\llbracket\alpha\rrbracket_\nu$.

Next we prove completeness. Suppose that the entailment of $\alpha$ by $\Gamma$ is valid in all orthomodular lattices. As noted after the definition of $H$, this translation preserves the value of each NQ formula under the corresponding valuations. By the completeness of GOM~\cite{Nish1980}, we then obtain a GOM derivation of $H(\alpha)$ from $H(\Gamma)$. By the normal sequent theorem (Theorem~2.5.a of~\cite{Nish1980}), there is a GOM derivation of $H(\Gamma)\vdash\Delta'$ in which all sequents are normal, for some $\Delta'\subseteq\{H(\alpha)\}$. If $\Delta'=\{H(\alpha)\}$, by Lemma~\ref{lem:GOM-to-NQ}, we obtain an NQ derivation of $J(H(\alpha))$ from $J(H(\Gamma))$. If $\Delta'=\varnothing$, by the same lemma, we obtain a derivation of $\bot$ from the same assumptions, and then $J(H(\alpha))$ by EFQ.

We verify that $\beta\simeq J(H(\beta))$ for every NQ formula $\beta$, by induction on its construction.

Base case. The translation fixes propositional variables. The image of $\bot$ is $q\wedge\neg q$, which is interderivable with $\bot$. From an assumption $q\wedge\neg q$, we obtain $q$ and $q\rightarrow\bot$, and hence $\bot$ by ${\rightarrow}\mathrm E$. The converse follows by EFQ.

Induction step. We proceed by cases according to the constructors.

\begin{enumerate}[(i)]
\item Case $\wedge$. By Lemma~\ref{lem:NQ-congruence} and the two induction hypotheses, we obtain the assertion.

\item Case $\rightarrow$. If $\beta=\gamma\rightarrow\delta$, we obtain interderivability of the corresponding expansions in conjunction and negation by the induction hypotheses and Lemma~\ref{lem:NQ-congruence}. The rules ${\rightarrow}\mathrm I_S$ and ${\rightarrow}\mathrm E_S$ relate the expansions to the primitive implications. Negation is included since $\neg\alpha$ abbreviates $\alpha\rightarrow\bot$.
\end{enumerate}

Finally, for each assumption $\gamma\in\Gamma$, we replace the translated assumption $J(H(\gamma))$ by its derivation from $\gamma$, and then append the derivation of $\alpha$ from $J(H(\alpha))$. Thus we obtain $\Gamma\vdash_{\mathrm{NQ}}\alpha$.
\end{proof}

\subsection{Typed Lambda Calculus Based on Quantum Logic}
\label{sec:lambda}

In this subsection, we present the typed lambda calculus $\QuLam$ corresponding to NQ. We first define types inductively from countably many type variables and the type constant $\bot$.

\begin{definition}[Types]
The \emph{types} of $\QuLam$ are inductively defined as follows:
\begin{itemize}
\item \emph{Type variable.} If $\alpha$ is a type variable, then $\alpha$ is a type.
\item \emph{Bottom type.} $\bot$ is a type.
\item \emph{Product type.} If $\alpha$ and $\beta$ are types, then $(\alpha \times \beta)$ is a type.
\item \emph{Implication type.} If $\alpha$ and $\beta$ are types, then $(\alpha \rightarrow \beta)$ is a type.
\item Only strings obtained by finitely many applications of the above rules are types.
\end{itemize}
Parentheses are omitted if there is no ambiguity.
\end{definition}

\begin{definition}[Typing environments]
A \emph{declaration} is an expression of the form $x : \alpha$ consisting of a term variable $x$ and a type $\alpha$.
An \emph{environment} is a finite partial map from term variables to types. We denote environments by $\Gamma$, or by $\Gamma_1$, $\Gamma_2$, $\ldots$. We write $\mathrm{dom}(\Gamma)$ for the domain of $\Gamma$. Two environments are \emph{consistent} when they agree on their common domain. In that case, $\Gamma_1\cup\Gamma_2$ denotes their union as partial maps.
\end{definition}

The environment specifies the available term variables and their types. We now present the formation rules for terms.

\begin{definition}[Typed $\lambda$-terms]\label{def:typed-lambda-terms}
\emph{Typed $\lambda$-terms} (or simply \emph{$\lambda$-terms}) are inductively defined as follows:
\begin{itemize}
\item \emph{Variable.} $x : \alpha$ is a $\lambda$-term under $\Gamma$ if $x : \alpha \in \Gamma$.
\item \emph{Ex falso.} $\varepsilon(M) : \alpha$ is a $\lambda$-term under $\Gamma$ if $M : \bot$ is a
$\lambda$-term under $\Gamma$, where $\alpha$ is an arbitrary type.
\item \emph{Projections.} $\pi_1(M) : \alpha$ and $\pi_2(M) : \beta$ are $\lambda$-terms under $\Gamma$
if $M : \alpha \times \beta$ is a $\lambda$-term under $\Gamma$.
\item \emph{Pairing.} $\langle M, N \rangle : \alpha \times \beta$ is a $\lambda$-term under $\Gamma_1 \cup \Gamma_2$ if $M : \alpha$ is a $\lambda$-term under $\Gamma_1$ and $N : \beta$ is a $\lambda$-term under $\Gamma_2$, where $\Gamma_1$ and $\Gamma_2$ are consistent.
\item \emph{Abstraction.} $(\lambda x.M) : \alpha \rightarrow \beta$ is a $\lambda$-term under $\Gamma$
if $M : \beta$ is a $\lambda$-term under $\{x : \alpha\}$.\footnote{Here the environment $\Gamma$ in the conclusion is understood as an ambient environment available after weakening. The premise for the body remains the singleton environment $\{x : \alpha\}$, so the abstraction is formed from a body whose available variable is $x$.}
\item \emph{Application.} $(MN) : \beta$ is a $\lambda$-term under $\Gamma_1 \cup \Gamma_2$ if
$M : \alpha \rightarrow \beta$ is a $\lambda$-term under $\Gamma_1$ and $N : \alpha$ is a $\lambda$-term under $\Gamma_2$, where $\Gamma_1$ and $\Gamma_2$ are consistent.
\item \emph{Modus tollens.} $\tau(M, N) : \alpha \rightarrow \bot$ is a $\lambda$-term under $\Gamma$ if
$M : \alpha \rightarrow \beta$ is a $\lambda$-term under $\varnothing$ and $N : \beta \rightarrow \bot$ is a $\lambda$-term under $\Gamma$.
\item \emph{Double-negation elimination.} $\eta(M) : \alpha$ is a $\lambda$-term under $\Gamma$ if
$M : (\alpha \rightarrow \bot) \rightarrow \bot$ is a $\lambda$-term under $\Gamma$.
\item \emph{Double-negation introduction.} $\theta(M) : (\alpha \rightarrow \bot) \rightarrow \bot$ is a $\lambda$-term under $\Gamma$ if
$M : \alpha$ is a $\lambda$-term under $\Gamma$.
\item \emph{Sasaki hook introduction.} $\zeta(M) : \alpha\rightarrow\beta$ is a $\lambda$-term under $\Gamma$ if $M : \neg(\alpha\times\neg(\alpha\times\beta))$ is a $\lambda$-term under $\Gamma$, with $\neg\alpha$ abbreviating $\alpha\rightarrow\bot$.
\item \emph{Sasaki hook elimination.} $\xi(M) : \neg(\alpha\times\neg(\alpha\times\beta))$ is a $\lambda$-term under $\Gamma$ if $M : \alpha\rightarrow\beta$ is a $\lambda$-term under $\Gamma$.
\item Only strings obtained by finitely many applications of the above rules are $\lambda$-terms.
\end{itemize}
\end{definition}

The correspondence between term constructors and inference rules of NQ is as follows.

\begin{center}
\begin{tabular}{ll}
\toprule
Term constructor & Inference rule \\
\midrule
$x$ & ASM \\
$\varepsilon(M)$ & EFQ \\
$\pi_1(M)$ & $\wedge\mathrm E_l$ \\
$\pi_2(M)$ & $\wedge\mathrm E_r$ \\
$\langle M,N\rangle$ & $\wedge\mathrm I$ \\
$\lambda x.M$ & ${\rightarrow}\mathrm I$ \\
$MN$ & ${\rightarrow}\mathrm E$ \\
$\tau(M,N)$ & MT \\
$\eta(M)$ & $\neg\neg\mathrm E$ \\
$\theta(M)$ & $\neg\neg\mathrm I$ \\
$\zeta(M)$ & ${\rightarrow}\mathrm I_S$ \\
$\xi(M)$ & ${\rightarrow}\mathrm E_S$ \\
\bottomrule
\end{tabular}
\end{center}

A \emph{typing judgment} is an expression of the form $\Gamma \vdash M : \alpha$, meaning that ``$M : \alpha$ is a $\lambda$-term under the environment $\Gamma$.''
In a typing judgment $\Gamma \vdash M : \alpha$, we often simply call $M$ a $\lambda$-term.
When $\Gamma$ is empty, we write $\vdash M : \alpha$ instead of $\varnothing \vdash M : \alpha$. We may omit braces around environments when listing their declarations.

\begin{definition}[Variables in terms]
\emph{Free} and \emph{bound variables} are defined in the usual way.
The set of free variables $\FV(M)$ of a term $M$ is given by the following clauses:

\begin{itemize}
\item $\FV(x) = \{x\}$.
\item $\FV(\varepsilon(M)) = \FV(M)$.
\item $\FV(\pi_i(M)) = \FV(M)$ $(i = 1, 2)$.
\item $\FV(\langle M, N \rangle) = \FV(M) \cup \FV(N)$.
\item $\FV(\lambda x.M) = \FV(M) \setminus \{x\}$.
\item $\FV(MN) = \FV(M) \cup \FV(N)$.
\item $\FV(\tau(M,N)) = \FV(M) \cup \FV(N)$.\footnote{Since $M$ is closed, this is in fact equal to $\FV(N)$.}
\item $\FV(\eta(M)) = \FV(M)$.
\item $\FV(\theta(M)) = \FV(M)$.
\item $\FV(\zeta(M))=\FV(M)$ and $\FV(\xi(M))=\FV(M)$.
\end{itemize}

In this definition, $\lambda$-terms equivalent under $\alpha$-conversion are treated as the same term.
\end{definition}

The following lemmas concern the behavior of environments under the typing rules. We use these lemmas in the proof of subject reduction and in the construction of abstractions over product arguments in Section~\ref{sec:product-abstraction}.

\begin{lemma}[Free variables]\label{lem:free-var}
If $\Gamma\vdash M:\alpha$, then $\FV(M)\subseteq\mathrm{dom}(\Gamma)$. Moreover, if a term of the form $\lambda x.N$ is typable, then it is closed and $\FV(N)\subseteq\{x\}$.
\end{lemma}
\begin{proof}
We proceed by induction on the typing derivation.

Base case. If the last rule types a variable $x$ under $\Gamma$, then $x\in\mathrm{dom}(\Gamma)$ and $\FV(x)=\{x\}$.

Induction step. We proceed by cases according to the last rule applied.\par
\begin{enumerate}[(i)]
\item Case pairing or application. The set of free variables of the resulting term is contained in the domain of the union of the premise environments.
\item Case abstraction. The premise has environment $\{x:\alpha\}$. Therefore, by the induction hypothesis, we obtain $\FV(N)\subseteq\{x\}$ and $\FV(\lambda x.N)=\varnothing$.
\item Case modus tollens. The first argument of $\tau$ is typed in the empty environment, and its second argument satisfies the induction hypothesis.
\item Case remaining constructors. For each unary constructor, including $\zeta$ and $\xi$, the free variables are those of the premise term. The induction hypothesis applies.\qedhere
\end{enumerate}
\end{proof}

We may also enlarge an environment without changing the type of a term.

\begin{lemma}[Weakening]\label{lem:weakening}
If $\Gamma\vdash M:\alpha$ and $\Delta$ is a consistent extension of $\Gamma$, then $\Delta\vdash M:\alpha$.
\end{lemma}
\begin{proof}
We proceed by induction on the typing derivation.

Base case. If $M$ is a variable typed under $\Gamma$, its declaration also belongs to $\Delta$.

Induction step. We proceed by cases according to the form of the last typing rule.\par
\begin{enumerate}[(i)]
\item Case pairing or application. By the induction hypotheses, both premises are typed under $\Delta$. Their union is $\Delta$, which is consistent with itself.
\item Case abstraction. The singleton premise is unchanged, since the conclusion may have any ambient environment.
\item Case modus tollens. The closed first premise is unchanged, and the induction hypothesis applies to the second premise.
\item Case remaining constructors. By the induction hypothesis, typing under $\Delta$ is preserved by their formation rules, including the rules for $\zeta$ and $\xi$.\qedhere
\end{enumerate}
\end{proof}

We next prove a substitution lemma in which all substituted terms are typed under a common environment. Since abstractions are closed, they remain unchanged under such substitutions.

\begin{lemma}[Substitution]\label{lem:substitution}
Suppose $\{x_1:\alpha_1,\ldots,x_k:\alpha_k\}\vdash M:\beta$ with distinct variables $x_i$, and $\Delta\vdash N_i:\alpha_i$ for every $i$. Then $\Delta\vdash M[x_1:=N_1,\ldots,x_k:=N_k]:\beta$, where substitution is simultaneous and avoids variable capture. The same statement holds when $M$ is typed under a subenvironment of $\{x_1:\alpha_1,\ldots,x_k:\alpha_k\}$.
\end{lemma}
\begin{proof}
We proceed by induction on a typing derivation of $M$.

Base case. If $M=x_i$, we obtain the required typing from $\Delta\vdash N_i:\alpha_i$.

Induction step. We proceed by cases according to the constructors.\par
\begin{enumerate}[(i)]
\item Case pairing or application. By the induction hypotheses, we obtain both premises under $\Delta$. We use their consistent union $\Delta\cup\Delta$.
\item Case abstraction. By Lemma~\ref{lem:free-var}, the abstraction is closed. After renaming the bound variable to avoid capture, substitution leaves it unchanged. By Lemma~\ref{lem:weakening}, we obtain its original typing under $\Delta$.
\item Case modus tollens. The first argument of $\tau$ is closed and unchanged by substitution. We apply the induction hypothesis to the second argument.
\item Case remaining constructors. The rules for $\varepsilon$, projections, $\eta$, $\theta$, $\zeta$, and $\xi$ follow by applying the induction hypotheses to their premises.
\end{enumerate}

An initial subenvironment can first be weakened to $\{x_1:\alpha_1,\ldots,x_k:\alpha_k\}$ using Lemma~\ref{lem:weakening}.
\end{proof}

We next define subterms in preparation for the conversion rules.

\begin{definition}[Subterms]
A $\lambda$-term $N$ that occurs as a subexpression of the $\lambda$-term $M$ is called a \emph{subterm} of $M$.
Equivalently, $N$ corresponds to a subtree of the abstract syntax tree of $M$.
\end{definition}

We now define conversion rules on $\lambda$-terms, corresponding on the proof-theoretic side to the elimination of detours in proofs of NQ. Type annotations in the conversion patterns and in Proposition~\ref{prop:dn-retraction} indicate the types assigned to the respective terms.

\begin{definition}[Conversion]\label{Conv}
A $\lambda$-term $M$ is said to be \emph{converted to $M'$ in one step}, written $M \Rightarrow M'$, if $M'$ is obtained from $M$ by one of the following replacements of a subterm:
\begin{itemize}
\item $\pi_1(\langle N:\alpha, K:\beta \rangle)$ is replaced with $N:\alpha$.
\item $\pi_2(\langle N:\alpha, K:\beta \rangle)$ is replaced with $K:\beta$.
\item $(\lambda x:\alpha. N:\beta)(K:\alpha)$ is replaced with $(N[x:=K]):\beta$. Here $N[x:=K]$ denotes the result of substituting $K$ for all free occurrences of $x$ in $N$. Bound variables are assumed to have been renamed consistently so as to avoid capturing free variables.
\item $\eta(\theta(N:\alpha))$ is replaced with $N:\alpha$.
\end{itemize}
\end{definition}

No conversion rule has $\zeta$ or $\xi$ at its root. In particular, the terms $\xi(\zeta(M))$ and $\zeta(\xi(N))$ remain unchanged at their roots under conversion.

We write $\Rightarrow^{*}$ for the reflexive and transitive closure of $\Rightarrow$. That is, for $\lambda$-terms $M$ and $N$, we write $M \Rightarrow^{*} N$ if there exist $\lambda$-terms $M_0, M_1, \ldots, M_n$ such that $M_0 = M$, $M_n = N$, and $M_i \Rightarrow M_{i+1}$ for all $0 \leq i < n$, where $n \geq 0$. The case $n=0$ means that no conversion step is performed.

\begin{definition}[Normal forms]
A $\lambda$-term to which no one-step conversion can be applied is called a \emph{normal form}. A $\lambda$-term $M$ is said to be \emph{normalizable} if there exists a normal form $N$ such that $M \Rightarrow^{*} N$.
\end{definition}

The conversion rules preserve typing, as we now verify.

\begin{proposition}[Subject reduction]\label{prop:subject-reduction}
If $\Gamma\vdash M:\alpha$ and $M\Rightarrow N$, then $\Gamma\vdash N:\alpha$. Consequently, if $\Gamma\vdash M:\alpha$ and $M\Rightarrow^{*}N$, then $\Gamma\vdash N:\alpha$.
\end{proposition}

\begin{proof}
We proceed by induction on the typing derivation.

Base case. A variable has no conversion step, so there is nothing to prove.

Induction step. We consider root conversions and conversions in proper subterms separately.\par
\begin{enumerate}[(i)]
\item Case projection. In a projection from a pair, we select a component typed under one of two consistent environments. By Lemma~\ref{lem:weakening}, that component is typed under their union.
\item Case beta conversion. In $(\lambda x.M)N\Rightarrow M[x:=N]$, we obtain $\{x:\alpha\}\vdash M:\beta$ from the abstraction rule and a typing of $N$ with type $\alpha$ from the application rule. By Lemma~\ref{lem:substitution}, the result is typed under the environment of $N$. By weakening, this typing is extended to the environment of the application.
\item Case double negation. In $\eta(\theta(M))\Rightarrow M$, $M$ is already typed with the required type in the premise.
\item Case proper subterm. By the induction hypothesis, the term in the affected premise remains typable with the same type. We then reconstruct the typing of the containing term by the formation rule. In particular, the singleton premise of abstraction and the closed premise of $\tau$ are still typed under their required environments. The same reasoning applies to $\zeta$ and $\xi$, since neither has a root conversion. By iterating subject reduction for one-step conversions, we obtain subject reduction for $\Rightarrow^{*}$.\qedhere
\end{enumerate}
\end{proof}

In what follows, we establish strong normalization using reducibility sets. The definitions and their basic properties are extensions of the method of Section~5 of~\cite{Toku2022} to the present term constructors. The proof method is based on~\cite{Gira1989}.

\begin{definition}[Reducibility sets]\label{def:reducibility}
We consider \emph{raw terms} formed from the constructors in Definition~\ref{def:typed-lambda-terms} without imposing typing conditions. Substitution is assumed to avoid variable capture. By removing the type annotations from the two projection conversions, beta conversion, and double-negation conversion of Definition~\ref{Conv}, we extend these conversions to raw terms. We write $\mathrm{SN}$ for the set of raw terms having no infinite conversion sequence.

For each type $\alpha$, we define the \emph{reducibility set} $\mathrm{RED}_\alpha$ of raw terms by induction on the structure of $\alpha$. If $\alpha$ is a type variable or $\bot$, we set $\mathrm{RED}_\alpha=\mathrm{SN}$. We define $\mathrm{RED}_{\alpha\times\beta}$ to consist of the terms $M$ such that $\pi_1(M)\in\mathrm{RED}_\alpha$ and $\pi_2(M)\in\mathrm{RED}_\beta$. We define $\mathrm{RED}_{\alpha\rightarrow\beta}$ to consist of the terms $M$ such that $MN\in\mathrm{RED}_\beta$ for every $N\in\mathrm{RED}_\alpha$. If $\alpha=\delta\rightarrow\bot$ and $\beta=\bot$, we impose the additional condition $\eta(M)\in\mathrm{RED}_\delta$. This condition corresponds to the argument type of the constructor $\eta$. We call the members of $\mathrm{RED}_\alpha$ \emph{reducible} at type $\alpha$.
\end{definition}

Every raw term has finitely many subterms, and only finitely many conversion rules can be applied to each subterm. Consequently, each raw term has only finitely many immediate reducts. Moreover, if $M\in\mathrm{SN}$, no infinite conversion sequence starts from $M$. Therefore, by König's lemma, the lengths of conversion sequences starting from $M$ have a finite maximum. We denote this maximum by $h(M)$ and use it as a measure in the following inductions.

\begin{lemma}[Properties of reducibility sets]\label{lem:reducibility-properties}
For every type $\alpha$, the following properties hold.
\begin{itemize}
\item (C1) Every term in $\mathrm{RED}_\alpha$ is strongly normalizing.
\item (C2) If $M\in\mathrm{RED}_\alpha$ and $M\Rightarrow M'$, then $M'\in\mathrm{RED}_\alpha$.
\item (C3) If the outermost constructor of $M$ is neither pairing, abstraction, nor $\theta$, and all immediate reducts of $M$ belong to $\mathrm{RED}_\alpha$, then $M\in\mathrm{RED}_\alpha$.
\end{itemize}
\end{lemma}
\begin{proof}
We prove the three properties simultaneously by induction on the complexity of $\alpha$.

Base case. If $\alpha$ is a type variable or $\bot$, then $\mathrm{RED}_\alpha=\mathrm{SN}$. The assertion in (C1) follows from the definition. For (C2), any reduct of a strongly normalizing term is strongly normalizing. For (C3), an infinite conversion sequence from $M$ would pass through an immediate reduct of $M$. Therefore, if all immediate reducts are strongly normalizing, then $M$ is strongly normalizing.

Induction step. We assume (C1)--(C3) for all types of lower complexity. We proceed by cases according to the form of $\alpha$.

\begin{enumerate}[(i)]
\setlength{\parskip}{\medskipamount}
\item Case product type. Suppose $\alpha=\beta\times\gamma$. If $M\in\mathrm{RED}_{\beta\times\gamma}$, then $\pi_1(M)\in\mathrm{RED}_\beta\subseteq\mathrm{SN}$. Since an infinite conversion sequence from $M$ would induce one from $\pi_1(M)$, we prove (C1). If $M\Rightarrow M'$, then $\pi_i(M)\Rightarrow\pi_i(M')$ for $i=1,2$. By applying (C2) at the component types, we prove (C2) for the product. For (C3), suppose $M$ is not a pair and all its immediate reducts are reducible at $\beta\times\gamma$. The projections $\pi_i(M)$ have no root conversion, and their immediate reducts are the corresponding projections of immediate reducts of $M$. By (C3) at the component types, we obtain $\pi_1(M)\in\mathrm{RED}_\beta$ and $\pi_2(M)\in\mathrm{RED}_\gamma$. Hence $M\in\mathrm{RED}_{\beta\times\gamma}$.

\item Case implication type. Suppose $\alpha=\beta\rightarrow\gamma$. By (C3) at type $\beta$, a variable $x$ belongs to $\mathrm{RED}_\beta$. If $M\in\mathrm{RED}_{\beta\rightarrow\gamma}$, then $Mx\in\mathrm{RED}_\gamma\subseteq\mathrm{SN}$. Since every conversion in $M$ induces a conversion in $Mx$, we establish (C1). If $M\Rightarrow M'$, then $MN\Rightarrow M'N$ for each $N\in\mathrm{RED}_\beta$. By (C2) at type $\gamma$, we obtain $M'N\in\mathrm{RED}_\gamma$. In the special case $\beta=\delta\rightarrow\bot$ and $\gamma=\bot$, we also have $\eta(M)\Rightarrow\eta(M')$. By (C2) at type $\delta$, the additional condition is preserved. This proves (C2).

For (C3), suppose that $M$ has none of the three excluded outermost constructors and that every immediate reduct $M'$ belongs to $\mathrm{RED}_{\beta\rightarrow\gamma}$. Since $M$ is not an abstraction, $MN$ has no root conversion for any $N$. Fix $N\in\mathrm{RED}_\beta$. By (C1) at type $\beta$, $h(N)$ is finite. We establish $MN\in\mathrm{RED}_\gamma$ by induction on $h(N)$.

Base case. If $h(N)=0$, each immediate reduct of $MN$ has the form $M'N$ with $M\Rightarrow M'$. It belongs to $\mathrm{RED}_\gamma$ by the assumption on $M'$. Therefore, by (C3) at type $\gamma$, we obtain $MN\in\mathrm{RED}_\gamma$.

Induction step. Assume the claim for reducible arguments of smaller maximal conversion length. The immediate reducts of $MN$ have the forms $M'N$ with $M\Rightarrow M'$ and $MN'$ with $N\Rightarrow N'$. The first form belongs to $\mathrm{RED}_\gamma$ by the assumption on $M'$. In the second form, by (C2) at type $\beta$, we obtain $N'\in\mathrm{RED}_\beta$, and $h(N')<h(N)$. Therefore, by the induction hypothesis, we obtain $MN'\in\mathrm{RED}_\gamma$. By (C3) at type $\gamma$, we obtain $MN\in\mathrm{RED}_\gamma$.

Thus $MN\in\mathrm{RED}_\gamma$ for every $N\in\mathrm{RED}_\beta$. If $\beta=\delta\rightarrow\bot$ and $\gamma=\bot$, we must also establish $\eta(M)\in\mathrm{RED}_\delta$. Since $M$ is not a $\theta$ term, $\eta(M)$ has no root conversion. Each of its immediate reducts is $\eta(M')$ with $M\Rightarrow M'$. These reducts belong to $\mathrm{RED}_\delta$ by the additional condition on $M'$. Therefore, by (C3) at type $\delta$, we obtain $\eta(M)\in\mathrm{RED}_\delta$.\qedhere
\end{enumerate}
\end{proof}

The next lemma states the constructor properties used in the substitution argument. In particular, statement (iv) concerning abstraction does not assume that $M$ or $N$ is typable. This is because the reducibility sets in Definition~\ref{def:reducibility} are defined for raw terms.

\begin{lemma}[Reducibility of constructors]\label{lem:reducibility-constructors}
The following statements hold for all types $\alpha$ and $\beta$.
\begin{enumerate}[(i)]
\item $\varepsilon(M)\in\mathrm{RED}_\alpha$ if $M\in\mathrm{RED}_\bot$.
\item $\pi_1(M)\in\mathrm{RED}_\alpha$ and $\pi_2(M)\in\mathrm{RED}_\beta$ if $M\in\mathrm{RED}_{\alpha\times\beta}$.
\item $\langle M,N\rangle\in\mathrm{RED}_{\alpha\times\beta}$ if $M\in\mathrm{RED}_\alpha$ and $N\in\mathrm{RED}_\beta$.
\item $\lambda x.M\in\mathrm{RED}_{\alpha\rightarrow\beta}$ if $M[x:=N]\in\mathrm{RED}_\beta$ for every $N\in\mathrm{RED}_\alpha$.
\item $MN\in\mathrm{RED}_\beta$ if $M\in\mathrm{RED}_{\alpha\rightarrow\beta}$ and $N\in\mathrm{RED}_\alpha$.
\item $\tau(M,N)\in\mathrm{RED}_{\alpha\rightarrow\bot}$ if $M\in\mathrm{RED}_{\alpha\rightarrow\beta}$ and $N\in\mathrm{RED}_{\beta\rightarrow\bot}$.
\item $\eta(M)\in\mathrm{RED}_\alpha$ if $M\in\mathrm{RED}_{(\alpha\rightarrow\bot)\rightarrow\bot}$.
\item $\theta(M)\in\mathrm{RED}_{(\alpha\rightarrow\bot)\rightarrow\bot}$ if $M\in\mathrm{RED}_\alpha$.
\end{enumerate}
\end{lemma}
\begin{proof}
We treat the constructors in turn. Projections and application satisfy the required conditions directly by Definition~\ref{def:reducibility}. The condition for $\eta$ is the additional clause of the same definition.

\begin{enumerate}[(i)]
\setlength{\parskip}{\medskipamount}
\item[(i)] Case ex falso. Suppose $M\in\mathrm{RED}_\bot$. We prove $\varepsilon(M)\in\mathrm{RED}_\alpha$ by induction on $h(M)$.

Base case. If $M$ is normal, then $\varepsilon(M)$ has no immediate reduct. Since it satisfies the outermost constructor condition of (C3), we obtain the claim.

Induction step. Every immediate reduct of $\varepsilon(M)$ is $\varepsilon(M')$ for some $M\Rightarrow M'$. By (C2), $M'\in\mathrm{RED}_\bot$, and $h(M')<h(M)$. By the induction hypothesis, we obtain $\varepsilon(M')\in\mathrm{RED}_\alpha$. Therefore, by (C3), we obtain $\varepsilon(M)\in\mathrm{RED}_\alpha$.

\item[(iii)] Case pairing. Suppose $M\in\mathrm{RED}_\alpha$ and $N\in\mathrm{RED}_\beta$. We prove that both projections of $\langle M,N\rangle$ are reducible by induction on $h(M)+h(N)$.

Base case. If both terms are normal, the only immediate reduct of $\pi_1(\langle M,N\rangle)$ is $M$, and the only immediate reduct of $\pi_2(\langle M,N\rangle)$ is $N$. Therefore, by (C3) at the respective component types, we obtain both claims.

Induction step. An immediate reduct of $\pi_i(\langle M,N\rangle)$ is either the selected component or a projection of $\langle M',N\rangle$ with $M\Rightarrow M'$ or $\langle M,N'\rangle$ with $N\Rightarrow N'$. The selected component is reducible by assumption. By (C2), reducibility of the reduced component is preserved. Since the corresponding sum of maximal conversion lengths is smaller, the induction hypothesis applies to the other reducts. By (C3) at the component type, we obtain reducibility of each projection and hence of the pair.

\item[(iv)] Case abstraction. Suppose $M[x:=N]\in\mathrm{RED}_\beta$ for every $N\in\mathrm{RED}_\alpha$. By (C3), the variable $x$ is reducible at type $\alpha$. Taking $N=x$, we obtain $M=M[x:=x]\in\mathrm{RED}_\beta$. Hence $M$ is strongly normalizing by (C1). Fix $N\in\mathrm{RED}_\alpha$. We prove $(\lambda x.M)N\in\mathrm{RED}_\beta$ by induction on $h(M)+h(N)$.

Base case. If $M$ and $N$ are normal, the only immediate reduct of $(\lambda x.M)N$ is $M[x:=N]$, which belongs to $\mathrm{RED}_\beta$ by assumption. This application has an outermost constructor permitted by (C3). Thus $(\lambda x.M)N\in\mathrm{RED}_\beta$.

Induction step. Every immediate reduct of $(\lambda x.M)N$ has one of the forms $M[x:=N]$, $(\lambda x.M')N$ with $M\Rightarrow M'$, and $(\lambda x.M)N'$ with $N\Rightarrow N'$. The first is reducible by assumption. In the second, the hypothesis on substitution is preserved, since $M[x:=K]\Rightarrow M'[x:=K]$ for every $K\in\mathrm{RED}_\alpha$, and (C2) applies at type $\beta$. In the third, $N'\in\mathrm{RED}_\alpha$ by (C2). Since the latter two conversions decrease the induction measure, by the induction hypothesis, we obtain reducibility in each of these cases. We complete the induction by applying (C3) at type $\beta$.

If $\alpha=\delta\rightarrow\bot$ and $\beta=\bot$, we also need $\eta(\lambda x.M)\in\mathrm{RED}_\delta$. We prove this by induction on $h(M)$.

Base case. When $M$ is normal, $\eta(\lambda x.M)$ has no immediate reduct. Since its outermost constructor is $\eta$, (C3) at type $\delta$ applies.

Induction step. Every immediate reduct is $\eta(\lambda x.M')$ with $M\Rightarrow M'$. As observed above, the substitution hypothesis holds for $M'$, and $h(M')<h(M)$. By the induction hypothesis, we obtain reducibility of each immediate reduct. Therefore, (C3) at type $\delta$ applies once more.

\item[(vi)] Case modus tollens. Suppose $M\in\mathrm{RED}_{\alpha\rightarrow\beta}$ and $N\in\mathrm{RED}_{\beta\rightarrow\bot}$. By (C1), both terms are strongly normalizing. For every $K\in\mathrm{RED}_\alpha$, the term $\tau(M,N)K$ has no root conversion. Each conversion occurs inside $M$, $N$, or $K$. Since all three are strongly normalizing, $\tau(M,N)K$ is strongly normalizing and hence belongs to $\mathrm{RED}_\bot$.

If $\alpha=\delta\rightarrow\bot$, the result type of $\tau(M,N)$ is a double-negation type. We must also prove $\eta(\tau(M,N))\in\mathrm{RED}_\delta$. We use induction on $h(M)+h(N)$.

Base case. If $M$ and $N$ are normal, $\eta(\tau(M,N))$ has no immediate reduct. Therefore, by (C3) at type $\delta$, we obtain reducibility.

Induction step. Its immediate reducts are $\eta(\tau(M',N))$ with $M\Rightarrow M'$ and $\eta(\tau(M,N'))$ with $N\Rightarrow N'$. By (C2), reducibility of the arguments is preserved, and each conversion decreases the induction measure. Therefore, the induction hypothesis applies to all immediate reducts. By (C3) at type $\delta$, we obtain the claim.

\item[(viii)] Case double-negation introduction. Suppose $M\in\mathrm{RED}_\alpha$. For any $N\in\mathrm{RED}_{\alpha\rightarrow\bot}$, the term $\theta(M)N$ has no root conversion, and conversions occur only in $M$ or $N$. By (C1), both arguments are strongly normalizing. Therefore $\theta(M)N\in\mathrm{RED}_\bot$. We also establish $\eta(\theta(M))\in\mathrm{RED}_\alpha$ by induction on $h(M)$.

Base case. If $M$ is normal, the only immediate reduct of $\eta(\theta(M))$ is $M$, which belongs to $\mathrm{RED}_\alpha$. Therefore, by (C3) at type $\alpha$, we obtain reducibility.

Induction step. The immediate reducts are $M$ and the terms $\eta(\theta(M'))$ with $M\Rightarrow M'$. The first is reducible by assumption. By (C2), we obtain $M'\in\mathrm{RED}_\alpha$, and $h(M')<h(M)$. By the induction hypothesis, we obtain reducibility of all remaining immediate reducts. Therefore, we conclude by (C3) at type $\alpha$.\qedhere
\end{enumerate}
\end{proof}

We can now establish strong normalization using simultaneous substitution and the constructor properties.

\begin{proposition}[Strong normalization]\label{prop:QuLam-SN}
Every well-typed term of QuLam is strongly normalizing under the conversions of Definition~\ref{Conv}.
\end{proposition}
\begin{proof}
We prove a stronger substitution statement by induction on a typing derivation. Suppose $\{x_1:\alpha_1,\ldots,x_k:\alpha_k\}\vdash M:\gamma$ and $N_i\in\mathrm{RED}_{\alpha_i}$ for every $i$. We set $\sigma=[x_1:=N_1,\ldots,x_k:=N_k]$ and prove $M\sigma\in\mathrm{RED}_\gamma$.

Base case. If $M=x_i$, then $M\sigma=N_i$, which is reducible by assumption.

Induction step. We proceed by cases according to the form of the last typing rule.

\begin{enumerate}[(i)]
\setlength{\parskip}{\medskipamount}
\item Case ex falso. If $M=\varepsilon(N)$, by the induction hypothesis, we obtain $N\sigma\in\mathrm{RED}_\bot$. The claim follows from part (i) of Lemma~\ref{lem:reducibility-constructors}.

\item Case projections. By the induction hypothesis, we obtain reducibility of the pair term after substitution. Part (ii) of Lemma~\ref{lem:reducibility-constructors} applies to the indicated projection.

\item Case pairing. By the induction hypotheses, we obtain reducibility of both components under the restrictions of $\sigma$ to their respective environments. Part (iii) of Lemma~\ref{lem:reducibility-constructors} applies to their pair.

\item Case abstraction. The body $N$ is typed under $\{x:\alpha\}$. For any $K\in\mathrm{RED}_\alpha$, by the induction hypothesis with $[x:=K]$, we obtain $N[x:=K]\in\mathrm{RED}_\beta$. Since the abstraction $\lambda x.N$ is closed, $M\sigma=\lambda x.N$ up to renaming of bound variables. By part (iv) of Lemma~\ref{lem:reducibility-constructors}, we obtain the conclusion.

\item Case application. By the induction hypotheses, we obtain reducibility of the function and argument after substitution. Part (v) of Lemma~\ref{lem:reducibility-constructors} applies.

\item Case modus tollens. The first argument of $\tau$ is typed in the empty environment and is unchanged by $\sigma$. By the induction hypotheses, we obtain reducibility of both arguments. By part (vi) of Lemma~\ref{lem:reducibility-constructors}, we obtain reducibility of their $\tau$ term.

\item Case double-negation elimination. By the induction hypothesis, we obtain a reducible argument of double-negation type. Part (vii) of Lemma~\ref{lem:reducibility-constructors} applies.

\item Case double-negation introduction. By the induction hypothesis, we obtain a reducible argument. Part (viii) of Lemma~\ref{lem:reducibility-constructors} applies.

\item Case Sasaki hook introduction. Suppose $M=\zeta(N)$ and the conclusion type is $\gamma=\alpha\rightarrow\beta$. By the induction hypothesis, we obtain $N\sigma\in\mathrm{RED}_{\neg(\alpha\times\neg(\alpha\times\beta))}$. We prove $\zeta(N\sigma)\in\mathrm{RED}_{\alpha\rightarrow\beta}$ by an inner induction on $h(N\sigma)$.

Base case. If $N\sigma$ is normal, $\zeta(N\sigma)$ has no immediate reduct. Since its outermost constructor satisfies the condition of (C3), it belongs to the required reducibility set.

Induction step. Every immediate reduct is $\zeta(N')$ with $N\sigma\Rightarrow N'$. By (C2), $N'$ remains reducible at its original type, and $h(N')<h(N\sigma)$. By the inner induction hypothesis, we obtain $\zeta(N')\in\mathrm{RED}_{\alpha\rightarrow\beta}$. We complete this case by applying (C3).

\item Case Sasaki hook elimination. Suppose $M=\xi(N)$ and the conclusion type is $\gamma=\neg(\alpha\times\neg(\alpha\times\beta))$. By the induction hypothesis, we obtain $N\sigma\in\mathrm{RED}_{\alpha\rightarrow\beta}$. We prove that $\xi(N\sigma)$ is reducible at the conclusion type by an inner induction on $h(N\sigma)$.

Base case. If $N\sigma$ is normal, $\xi(N\sigma)$ has no immediate reduct, and (C3) applies.

Induction step. Every immediate reduct is $\xi(N')$ with $N\sigma\Rightarrow N'$. By (C2), the reduct $N'$ is reducible and has smaller maximal conversion length. The inner induction hypothesis applies to $\xi(N')$. By (C3), we obtain the conclusion.
\end{enumerate}

By (C3), every variable belongs to its corresponding reducibility set. Taking $N_i=x_i$ in the substitution statement, we obtain $M\in\mathrm{RED}_\gamma$. By (C1), we obtain strong normalization.
\end{proof}

\section{Curry--Howard Correspondence}
\label{sec:CH}

In this section, we relate NQ derivations to typing derivations in $\QuLam$ through the Curry--Howard correspondence~\cite{Gira1989,Howa1980,Sore2006}. We follow the term assignment in~\cite{Toku2022} for the rules shared with that work. We specify the clauses corresponding to ${\rightarrow}\mathrm I_S$ and ${\rightarrow}\mathrm E_S$ in terms of the constructors $\zeta$ and $\xi$, respectively.

\subsection{Types as Propositions}
\label{sec:typeasprop}

We first describe the correspondence between formulas of NQ and types of $\QuLam$. Let $\mathrm{PropVar}$ be the set of propositional variables and $\mathrm{TypeVar}$ the set of type variables. Since both sets are countably infinite, we can fix a bijection $\Phi:\mathrm{PropVar}\to\mathrm{TypeVar}$. Let $\mathrm{Form}$ be the set of all formulas of NQ and $\mathrm{Type}$ the set of all types of $\QuLam$. We extend $\Phi$ by structural recursion to obtain mutually inverse maps $\Phi:\mathrm{Form}\to\mathrm{Type}$ and $\Phi^{-1}:\mathrm{Type}\to\mathrm{Form}$ that preserve the structure of formulas and types as follows:

\begin{itemize}
\item $\Phi(\bot):=\bot$.
\item $\Phi(\alpha\wedge\beta):=\Phi(\alpha)\times\Phi(\beta)$.
\item $\Phi(\alpha\rightarrow\beta):=\Phi(\alpha)\rightarrow\Phi(\beta)$.
\end{itemize}
\begin{itemize}
\item For each type variable $\alpha$, $\Phi^{-1}(\alpha)$ is the corresponding propositional variable.
\item $\Phi^{-1}(\bot):=\bot$.
\item $\Phi^{-1}(\alpha\times\beta):=\Phi^{-1}(\alpha)\wedge\Phi^{-1}(\beta)$.
\item $\Phi^{-1}(\alpha\rightarrow\beta):=\Phi^{-1}(\alpha)\rightarrow\Phi^{-1}(\beta)$.
\end{itemize}

\subsection{Terms as Proofs}
\label{sec:termasproof}

We next describe the correspondence between proofs in NQ and terms in $\QuLam$. Let $\mathrm{Label}$ be the set of labels used for assumptions in NQ, and let $\mathrm{Var}$ be the set of term variables in $\QuLam$. Since both sets are countably infinite, we fix a bijection $\Psi:\mathrm{Label}\to\mathrm{Var}$. We use the same symbol $\Psi$ for the recursively defined map from labeled NQ derivations to QuLam terms. For a finite set $\Gamma$ of labeled assumptions, we write $\Psi(\Gamma)$ for the environment $\{\Psi(u):\Phi(\beta)\mid u:\beta\in\Gamma\}$.

We regard two labeled NQ derivations as equivalent if they differ only in the names of discharged assumption labels. We also regard two QuLam typing derivations as equivalent if they differ only in the names of bound term variables.

We define a \emph{canonical typing derivation} as a derivation obtained recursively from a labeled NQ derivation by assigning the corresponding QuLam rule to each NQ rule. When constructing a canonical QuLam typing derivation, we associate each NQ subderivation with a QuLam typing judgment. The environment of this typing judgment consists of the declarations corresponding to the undischarged assumptions of the NQ subderivation, except for the extensions described below.

The implication introduction rule ${\rightarrow}\mathrm I$ corresponds to $\lambda$-abstraction. In NQ, a derivation of $\beta$ used for this rule cannot have any undischarged assumptions other than the labeled assumption $u:\alpha$. The corresponding QuLam derivation has the singleton environment $\{x:\Phi(\alpha)\}$, where $x=\Psi(u)$, and assigns type $\Phi(\beta)$ to the body $M$. The result is the closed term $\lambda x.M:\Phi(\alpha)\rightarrow\Phi(\beta)$. The rule ${\rightarrow}\mathrm I_S$ also introduces an implication, but it does not discharge any assumptions.

If the implication introduction rule ${\rightarrow}\mathrm I$ discharges an assumption $u:\alpha$ that is not used in its premise derivation, the corresponding QuLam term $M$ is initially typed under the empty environment. However, the QuLam abstraction rule requires $M$ to be typed under the environment $\{x:\Phi(\alpha)\}$, where $x=\Psi(u)$. We therefore use Lemma~\ref{lem:weakening} to extend the environment of the typing derivation of $M$. This operation may also extend the environments of intermediate typing judgments in that derivation. In canonical typing derivations, we allow only the extensions required by this procedure.

The rules concerning negation and double negation correspond to the term constructors $\tau$, $\eta$, and $\theta$. The modus tollens rule {\rm MT} of NQ combines a closed proof of $\alpha\rightarrow\beta$ with a proof of $\beta\rightarrow\bot$ under assumptions $\Gamma$ to derive $\alpha\rightarrow\bot$ under the same assumptions $\Gamma$. In $\QuLam$, this inference is represented by the constructor $\tau$. More precisely, given a closed term $M:\Phi(\alpha)\rightarrow\Phi(\beta)$ and a term $N:\Phi(\beta)\rightarrow\bot$ under $\Psi(\Gamma)$, we obtain a term $\tau(M,N):\Phi(\alpha)\rightarrow\bot$ under $\Psi(\Gamma)$.

The double-negation introduction and elimination rules ${\neg\neg}{\rm I}$ and ${\neg\neg}{\rm E}$ correspond to the term constructors $\theta$ and $\eta$, respectively. For example, by applying the double-negation elimination rule ${\neg\neg}{\rm E}$ to a proof $\mathcal D$ of $\neg\neg\alpha$ under the assumptions $\Gamma$, we obtain a proof of $\alpha$ under $\Gamma$. In $\QuLam$, this situation is reflected by the typing judgments $\Psi(\Gamma)\vdash\Psi(\mathcal D):(\Phi(\alpha)\rightarrow\bot)\rightarrow\bot$ and $\Psi(\Gamma)\vdash\eta(\Psi(\mathcal D)):\Phi(\alpha)$.

\begin{proposition}[Derivation correspondence]\label{prop:derivation-correspondence}
When we identify equivalent derivations, the map that sends each labeled NQ derivation to its canonical QuLam typing derivation is a bijection. If an NQ derivation $\mathcal D$ of $\alpha$ has undischarged assumptions $\Gamma$, then the corresponding typing derivation concludes with $\Psi(\Gamma)\vdash\Psi(\mathcal D):\Phi(\alpha)$.
\end{proposition}
\begin{proof}
We prove that the two transformations are inverses of each other. We use induction on the height of an NQ derivation in one direction and on the height of a canonical QuLam typing derivation in the other.

Base case. An NQ assumption $u:\alpha$ corresponds to the QuLam variable $\Psi(u):\Phi(\alpha)$. Both transformations return the same labeled assumption or declaration.

Induction step. We proceed by cases according to the last rule applied. In each case, we apply the induction hypotheses to the premise derivations.

\begin{enumerate}[(i)]
\item Case implication introduction. In NQ, the premise derivation has no undischarged assumptions other than the assumption discharged by ${\rightarrow}\mathrm I$. In QuLam, the body of the corresponding abstraction is typed under an environment containing only the declaration corresponding to that assumption. If the assumption is not used in the NQ premise derivation, we use weakening to type the QuLam term under this environment. This restriction is preserved in both directions of the correspondence, and renaming the discharged NQ assumption label corresponds to renaming the bound variable in QuLam. When transforming the QuLam typing derivation back into an NQ derivation, we remove the environment extensions introduced when the discharged assumption is not used.

\item Case Sasaki hook. The rules ${\rightarrow}\mathrm I_S$ and ${\rightarrow}\mathrm E_S$ correspond to $\zeta$ and $\xi$, respectively. In particular, $\Psi(({\rightarrow}\mathrm I_S)(\mathcal D))=\zeta(\Psi(\mathcal D))$ and $\Psi(({\rightarrow}\mathrm E_S)(\mathcal D))=\xi(\Psi(\mathcal D))$. In the reverse transformation, we apply ${\rightarrow}\mathrm I_S$ when the outermost constructor of the QuLam term is $\zeta$ and ${\rightarrow}\mathrm E_S$ when it is $\xi$. Neither NQ rule discharges any assumptions.

\item Case remaining rules. The NQ rules EFQ, conjunction introduction, conjunction elimination, implication elimination, MT, double-negation introduction, and double-negation elimination correspond to the QuLam constructors $\varepsilon$, pairing, projections, application, $\tau$, $\theta$, and $\eta$, respectively. The first premise of MT has no undischarged assumptions, and the first premise of the QuLam typing rule for $\tau$ is typed under the empty environment.\qedhere
\end{enumerate}
\end{proof}

\subsection{Constructors for the Sasaki Hook}
\label{sec:sasaki-constructors}

The constructors $\zeta$ and $\xi$ correspond to the NQ rules ${\rightarrow}\mathrm I_S$ and ${\rightarrow}\mathrm E_S$, respectively. Their typing rules are $\Gamma\vdash M:\neg(\alpha\times\neg(\alpha\times\beta))\Longrightarrow\Gamma\vdash\zeta(M):\alpha\rightarrow\beta$ and $\Gamma\vdash N:\alpha\rightarrow\beta\Longrightarrow\Gamma\vdash\xi(N):\neg(\alpha\times\neg(\alpha\times\beta))$. In the typing rule for $\zeta$, the premise and conclusion have the same environment $\Gamma$. In contrast, the body of a $\lambda$-abstraction must be typed under an environment containing a single declaration.

\begin{example}[Term for the rule MP]\label{ex:orthomodular-term}
The closed term $\lambda z.\,(\zeta(\pi_2 z))(\pi_1 z)$ has type $(\alpha\times\neg(\alpha\times\neg(\alpha\times\beta)))\rightarrow\beta$ in QuLam. To type the body of this abstraction, we assume $z:\alpha\times\neg(\alpha\times\neg(\alpha\times\beta))$. Thus $\zeta(\pi_2 z):\alpha\rightarrow\beta$ and $\pi_1z:\alpha$, and we obtain a term of type $\beta$ by application. This term corresponds to the rule MP. In the proof of Lemma~\ref{lem:GOM-to-NQ}, MP is used in the derivation that replaces the O-modular rule.
\end{example}

\begin{example}[Implication from a conjunction]\label{ex:conjunction-implication}
Assume $z:\alpha\times\beta$. Since $\pi_2x:\neg(\alpha\times\beta)$ under the assumption $x:\alpha\times\neg(\alpha\times\beta)$, the closed term $\lambda x.\pi_2x$ has type $(\alpha\times\neg(\alpha\times\beta))\rightarrow\neg(\alpha\times\beta)$. Moreover, $\theta(z)$ has type $\neg\neg(\alpha\times\beta)$. By the typing rule for $\tau$, we obtain $\tau(\lambda x.\pi_2x,\theta(z)):\neg(\alpha\times\neg(\alpha\times\beta))$. Applying the typing rule for $\zeta$, we obtain $z:\alpha\times\beta\vdash\zeta(\tau(\lambda x.\pi_2x,\theta(z))):\alpha\rightarrow\beta$.
\end{example}

\section{Orthomodular Lattice Semantics}
\label{sec:semantics}

In this section, we present a semantics for typing judgments in $\QuLam$ using orthomodular lattices. We then prove soundness and completeness with respect to this semantics.

\subsection{Interpretation of Types}
\label{sec:semantics-types}

We interpret the types of $\QuLam$ in orthomodular lattices, as we do the formulas of NQ.

We now define the interpretation of types inductively from a valuation of type variables according to the structure of types.

\begin{definition}[Interpretation of types]\label{def:interpretation}
Let $L$ be an orthomodular lattice, and let $\rho:\mathrm{TypeVar}\rightarrow L$ be a valuation of type variables. We define an element $\llbracket\alpha\rrbracket_\rho\in L$ for every type $\alpha$ as follows:

\begin{itemize}
\item $\llbracket\alpha\rrbracket_\rho=\rho(\alpha)$ for $\alpha\in\mathrm{TypeVar}$.
\item $\llbracket\bot\rrbracket_\rho=0$.
\item $\llbracket\alpha\times\beta\rrbracket_\rho=\llbracket\alpha\rrbracket_\rho\wedge\llbracket\beta\rrbracket_\rho$.
\item $\llbracket\alpha\rightarrow\beta\rrbracket_\rho=\llbracket\alpha\rrbracket_\rho\rightarrow_S\llbracket\beta\rrbracket_\rho=\llbracket\alpha\rrbracket_\rho^\perp\vee(\llbracket\alpha\rrbracket_\rho\wedge\llbracket\beta\rrbracket_\rho)$.
\end{itemize}
\end{definition}

By this definition, for every type $\alpha$, the value $\llbracket\alpha\rrbracket_\rho$ is uniquely determined as an element of $L$, and we obtain a map $\llbracket-\rrbracket_\rho:\mathrm{Type}\rightarrow L$ that preserves the structure of types. In what follows, when the orthomodular lattice $L$ and the valuation $\rho$ of type variables are clear from the context, we omit the subscript and write $\llbracket\alpha\rrbracket$ instead of $\llbracket\alpha\rrbracket_\rho$.

Next, we prove that the interpretations of formulas and types coincide under corresponding valuations.

\begin{lemma}[Transport of valuations]\label{lem:transport}
Let $L$ be an orthomodular lattice, and $\rho$ a valuation of type variables. We put $\nu(p)=\rho(\Phi(p))$ for every propositional variable $p$. Then $\llbracket\alpha\rrbracket_\nu=\llbracket\Phi(\alpha)\rrbracket_\rho$ for every NQ formula $\alpha$. The same equality holds for the meet of the values of any finite set of corresponding assumptions.
\end{lemma}
\begin{proof}
We proceed by induction on the construction of $\alpha$.

Base case. The claim follows from the definition of $\nu$ for propositional variables. Both interpretations assign $0$ to $\bot$.

Induction step. We proceed by cases according to the form of $\alpha$.

\begin{enumerate}[(i)]
\item Case conjunction. Both interpretations use the meet. We apply the induction hypotheses to the respective operands.

\item Case implication. Both interpretations use the Sasaki hook. We apply the induction hypotheses to the respective operands.
\end{enumerate}

For finite sets, the equalities for their individual elements imply that the meets of the respective values are also equal. The empty set has value $1$ on both sides.
\end{proof}

\subsection{Semantics of Typing Judgments}
\label{sec:semantics-contexts}

In this subsection, we interpret environments and typing judgments with respect to the same orthomodular lattice $L$ and valuation $\rho$ of type variables as in the preceding subsection.

First, we define the interpretation of an environment. Since an environment $\Gamma$ is a finite partial map, it can be represented as the finite set $\{x_1:\alpha_1,\dots,x_n:\alpha_n\}$. We define its interpretation as the meet $\llbracket\Gamma\rrbracket_\rho=\llbracket\alpha_1\rrbracket_\rho\wedge\cdots\wedge\llbracket\alpha_n\rrbracket_\rho$. For the empty environment, we set $\llbracket\varnothing\rrbracket_\rho=1$ as the meet of the empty family. The interpretation $\llbracket\Gamma\rrbracket_\rho$ does not depend on the order in which the elements of $\Gamma$ are listed or on the names of the variables. Thus an environment is interpreted by an element $\llbracket\Gamma\rrbracket_\rho\in L$, which coincides with the interpretation, on the NQ side, of the conjunction of the assumptions corresponding to $\Gamma$.

Next, we interpret a typing judgment $\Gamma\vdash M:\alpha$ through the order relation between the interpretations of its environment and type. We say that the typing judgment $\Gamma\vdash M:\alpha$ is \emph{valid} in $(L,\rho)$ when the inequality $\llbracket\Gamma\rrbracket_\rho\leq\llbracket\alpha\rrbracket_\rho$ holds. In particular, since the interpretation of the empty environment is $1$, a closed typing judgment $\vdash M:\alpha$ is valid in $(L,\rho)$ if and only if $\llbracket\alpha\rrbracket_\rho=1$.

By Lemma~\ref{lem:transport}, this inequality coincides, through the inverse correspondence $\Phi^{-1}:\mathrm{Type}\rightarrow\mathrm{Form}$ of Section~\ref{sec:CH}, with the inequality holding for the corresponding NQ formulas under the same orthomodular lattice and the corresponding valuation.

\subsection{Soundness of Typing Judgments}
\label{sec:soundness}

In this subsection, we prove that $\QuLam$ is sound for the semantics introduced in Sections~\ref{sec:semantics-types} and~\ref{sec:semantics-contexts}.

\begin{theorem}[Soundness]\label{thm:soundness}
Let $L$ be an orthomodular lattice, and let $\rho:\mathrm{TypeVar}\to L$ be a valuation of type variables. If $\Gamma\vdash M:\alpha$ is derivable in QuLam, then $\llbracket\Gamma\rrbracket_\rho\leq\llbracket\alpha\rrbracket_\rho$. In particular, if $\vdash M:\alpha$ is derivable in QuLam, then $\llbracket\alpha\rrbracket_\rho=1$.
\end{theorem}
\begin{proof}
We proceed by induction on a QuLam typing derivation. All interpretations below are taken in the orthomodular lattice $L$ under the valuation $\rho$.

Base case. If the last rule types a variable $x:\alpha$ under $\Gamma$, then $x:\alpha\in\Gamma$. Consequently, $\llbracket\Gamma\rrbracket_\rho\leq\llbracket\alpha\rrbracket_\rho$.

Induction step. We proceed by cases according to the form of the last typing rule.

\begin{enumerate}[(i)]
\item Case ex falso. By the induction hypothesis, we obtain $\llbracket\Gamma\rrbracket_\rho\leq0$. Therefore, $\llbracket\Gamma\rrbracket_\rho=0\leq\llbracket\alpha\rrbracket_\rho$ holds for every conclusion type $\alpha$.

\item Case projections. Suppose the premise has type $\alpha\times\beta$. By the induction hypothesis, $\llbracket\Gamma\rrbracket_\rho\leq\llbracket\alpha\rrbracket_\rho\wedge\llbracket\beta\rrbracket_\rho$. The required inequality follows for either projection.

\item Case pairing. Let $\Gamma=\Gamma_1\cup\Gamma_2$, and suppose the two premise environments are consistent. Put $a=\llbracket\alpha\rrbracket_\rho$ and $b=\llbracket\beta\rrbracket_\rho$. Since $\llbracket\Gamma\rrbracket_\rho$ is below the interpretations of both premise environments, the induction hypotheses imply $\llbracket\Gamma\rrbracket_\rho\leq a\wedge b=\llbracket\alpha\times\beta\rrbracket_\rho$.

\item Case abstraction. The premise is typed under $\{x:\alpha\}$ and has conclusion type $\beta$. By the induction hypothesis, we have $a\leq b$, where $a=\llbracket\alpha\rrbracket_\rho$ and $b=\llbracket\beta\rrbracket_\rho$. Therefore, $a\rightarrow_S b=a^\perp\vee(a\wedge b)=1$. Moreover, the abstraction has type $\alpha\rightarrow\beta$ under any permitted ambient environment $\Gamma$, and $\llbracket\Gamma\rrbracket_\rho\leq1$ holds.

\item Case application. Suppose the function and argument have types $\alpha\rightarrow\beta$ and $\alpha$ under consistent environments $\Gamma_1$ and $\Gamma_2$, respectively. We write $c=\llbracket\Gamma_1\cup\Gamma_2\rrbracket_\rho$, $a=\llbracket\alpha\rrbracket_\rho$, and $b=\llbracket\beta\rrbracket_\rho$. By the induction hypotheses, we obtain $c\leq a\rightarrow_S b$ and $c\leq a$. Therefore, using the orthomodularity condition, we have $c\leq a\wedge(a\rightarrow_S b)=a\wedge b\leq b$.

\item Case modus tollens. From the closed first premise, we obtain $1\leq a\rightarrow_S b$ for the values $a,b$ of its antecedent and consequent types $\alpha$ and $\beta$. Since $\varphi_a(1)=a$, we obtain $a\leq b$ by Lemma~\ref{lem:sasaki-adjunction}. By the induction hypothesis for the second premise, we obtain $\llbracket\Gamma\rrbracket_\rho\leq b^\perp$. Therefore, $\llbracket\Gamma\rrbracket_\rho\leq a^\perp=\llbracket\neg\alpha\rrbracket_\rho$ holds.

\item Case double negation. Orthocomplementation satisfies $(a^\perp)^\perp=a$. Therefore, the interpretation of $\neg\neg\alpha$ equals that of $\alpha$, and the induction hypothesis applies to both $\eta$ and $\theta$.

\item Case Sasaki hook. In an orthomodular lattice, $\bigl(a\wedge(a\wedge b)^\perp\bigr)^\perp=a^\perp\vee(a\wedge b)=a\rightarrow_S b$. Therefore, the premise and conclusion types of each rule for $\zeta$ or $\xi$ have equal interpretations. By the induction hypothesis, we obtain the required inequality.
\end{enumerate}

These cases cover all typing rules. If $\Gamma=\varnothing$, then $\llbracket\Gamma\rrbracket_\rho=1$, and the final assertion follows.
\end{proof}

Next, we prove completeness of QuLam typing judgments using the completeness of NQ and Proposition~\ref{prop:derivation-correspondence}.

\begin{corollary}[Completeness of typing judgments]\label{cor:typing-completeness}
Suppose $\llbracket\Gamma\rrbracket_\rho\leq\llbracket\alpha\rrbracket_\rho$ holds for every orthomodular lattice $L$ and every valuation $\rho$ of type variables. Then $\Gamma\vdash M:\alpha$ is derivable for some QuLam term $M$.
\end{corollary}
\begin{proof}
Let $L$ be an arbitrary orthomodular lattice, and let $\nu$ be a valuation of propositional variables in $L$. Since $\Phi$ is a bijection on variables, we have $\nu(p)=\rho(\Phi(p))$ if we set $\rho=\nu\circ\Phi^{-1}$. Hence, by Lemma~\ref{lem:transport}, the assumed inequality also holds for the NQ formulas corresponding to $\Gamma$ and $\alpha$. Since $L$ and $\nu$ are arbitrary, this inequality and the completeness of NQ (Theorem~\ref{thm:NQ-corrected-completeness}) imply that the formula corresponding to $\alpha$ is derivable in NQ from the assumptions corresponding to $\Gamma$. Next, using Proposition~\ref{prop:derivation-correspondence}, we transform such an NQ derivation into a canonical typing derivation $\Gamma'\vdash M:\alpha$ with $\Gamma'\subseteq\Gamma$. Finally, by Lemma~\ref{lem:weakening}, we derive $\Gamma\vdash M:\alpha$.
\end{proof}

\section{Structure of Terms}
\label{sec:control}

In this section, we examine the structure of terms in $\QuLam$ from the standpoint of the proof theory of NQ. We first discuss negative inference, double negation, and their corresponding term constructions. We then consider restricted abstraction, the use of product arguments, and a deduction principle based on the Sasaki projection.

\subsection{Negative Composition}
\label{sec:tau}

The constructor $\tau$ implements the modus tollens rule of $\mathrm{NQ}$. Its first argument is a closed implication term, while its second argument carries the available assumptions.

\begin{proposition}[Term assignment for modus tollens]
\label{prop:tau-mt}
If $\vdash M:\alpha\rightarrow\beta$ and $\Gamma\vdash N:\beta\rightarrow\bot$, then $\Gamma\vdash\tau(M,N):\alpha\rightarrow\bot$. Under the correspondence of Proposition~\ref{prop:derivation-correspondence}, this typing rule serves as the term assignment for the modus tollens rule of NQ.
\end{proposition}

\begin{proof}
We obtain this typing judgment directly from the formation rule for $\tau$ in Definition~\ref{def:typed-lambda-terms}. The premise $\vdash M:\alpha\rightarrow\beta$ corresponds to a closed proof of $\Phi^{-1}(\alpha)\rightarrow\Phi^{-1}(\beta)$, and the premise $\Gamma\vdash N:\beta\rightarrow\bot$ corresponds to a proof of $\Phi^{-1}(\beta)\rightarrow\bot$ under the assumptions represented by $\Gamma$. Therefore, $\tau(M,N)$ is the term corresponding to the conclusion of modus tollens in $\mathrm{NQ}$.
\end{proof}

Under the hypotheses of Proposition~\ref{prop:tau-mt}, if $\Delta\vdash K:\alpha$ and $\Gamma$ and $\Delta$ are consistent, we obtain $\Gamma\cup\Delta\vdash\tau(M,N)K:\bot$ by the formation rule for application.

\subsection{Double Negation}
\label{sec:eta-theta}

The constructors $\eta$ and $\theta$ correspond to the double-negation elimination and introduction rules of $\mathrm{NQ}$, respectively. For a type $\alpha$, we use $\neg\alpha$ as an abbreviation for $\alpha\rightarrow\bot$ and $\neg\neg\alpha$ as an abbreviation for $(\alpha\rightarrow\bot)\rightarrow\bot$. The constructor $\theta$ maps a term of type $\alpha$ to a term of type $\neg\neg\alpha$, and $\eta$ maps a term of type $\neg\neg\alpha$ to a term of type $\alpha$.\footnote{In continuation-passing style (CPS), the type $(\alpha\rightarrow\kappa)\rightarrow\kappa$ represents a computation with answer type $\kappa$. When $\kappa=\bot$, this type is $\neg\neg\alpha$. A related type structure occurs in Parigot's $\lambda\mu$-calculus~\cite{Pari1992}. Connections between classical logic and control operators are also studied in~\cite{Ario2003,Grif1990}.}

\begin{proposition}[Double-negation retraction]
\label{prop:dn-retraction}
For every type $\alpha$, the terms $\theta_\alpha:=\lambda x:\alpha.\,\theta(x)$ and $\eta_\alpha:=\lambda y:(\alpha\rightarrow\bot)\rightarrow\bot.\,\eta(y)$ have types $\alpha\rightarrow((\alpha\rightarrow\bot)\rightarrow\bot)$ and $((\alpha\rightarrow\bot)\rightarrow\bot)\rightarrow\alpha$, respectively. Moreover, for every term $M:\alpha$, we have $\eta_\alpha(\theta_\alpha M)\Rightarrow^* M$.
\end{proposition}

\begin{proof}
We obtain the typing of $\theta_\alpha$ and $\eta_\alpha$ from the formation rules for $\theta$, $\eta$, and $\lambda$-abstraction. For the conversion claim, we obtain $\eta_\alpha(\theta_\alpha M)\Rightarrow^*\eta(\theta(M))$ by ordinary beta conversion and $\eta(\theta(M))\Rightarrow M$ by the double-negation conversion rule.
\end{proof}

The double-negation pair behaves as follows in $\QuLam$. The constructor $\theta$ embeds a term into its double-negation type, and $\eta$ returns it to its original type. The composite $\eta_\alpha\circ\theta_\alpha$ reduces to the identity on terms of the given type.

We can also combine the constructors $\tau$ and $\eta$ as follows. Suppose $\vdash M:\alpha\rightarrow\beta$ and $\vdash N:\alpha$, and let $x$ be a variable of type $\beta\rightarrow\bot$. Then $x:\beta\rightarrow\bot\vdash\tau(M,x)N:\bot$, and hence $\vdash\eta(\lambda x.\tau(M,x)N):\beta$. We obtain this typing from Proposition~\ref{prop:tau-mt} and the formation rules for application, abstraction, and double-negation elimination.

\subsection{Restrictions on Dependencies}
\label{sec:design}

The restriction on $\lambda$-abstraction governs the structure of terms in $\QuLam$. The implication introduction rule of $\mathrm{NQ}$ allows a proof of $\alpha\rightarrow\beta$ to be constructed from a derivation of $\beta$ under the single assumption $\alpha$. Under the Curry--Howard correspondence, we obtain the term rule by which $\lambda x.M$ is constructed from a derivation of $M:\beta$ under the singleton environment $\{x:\alpha\}$. Therefore, every $\lambda$-abstraction is a self-contained proof fragment whose body depends only on its explicit argument.

\begin{proposition}[Explicit dependency principle]
\label{prop:explicit-dependency}
Let $\lambda x.M$ be a well-typed abstraction of type $\alpha\rightarrow\beta$ in QuLam. Then $\FV(M)\subseteq\{x\}$ and $\FV(\lambda x.M)=\varnothing$. In particular, variables other than $x$ cannot occur free in $M$. Instead, we can combine variables from the surrounding environment into a single product argument and access its components through projections, as in Proposition~\ref{prop:tuple-abstraction}. Closed terms may also occur within the body.
\end{proposition}

\begin{proof}
Since $M:\beta$ is derived under the singleton environment $\{x:\alpha\}$ by the formation rule for $\lambda$-abstraction, the only possible free variable of $M$ is $x$. Therefore, by binding $x$, we obtain $\FV(\lambda x.M)=\varnothing$.
\end{proof}

The same principle also applies to $\tau$. In the formation rule for $\tau(M,N)$, the implication term $M:\alpha\rightarrow\beta$ is closed, and the environment of the conclusion is inherited from the negative premise $N:\beta\rightarrow\bot$. Therefore, the positive implication from $\alpha$ to $\beta$ is a self-contained term, while the surrounding assumptions are concentrated in the negative part of the construction. This expresses, at the level of terms, the corresponding restriction in the modus tollens rule of $\mathrm{NQ}$. In contrast, a term $\zeta(M)$ of implication type may have free variables inherited from its premise $M$, as permitted by Definition~\ref{def:typed-lambda-terms}.

\subsection{Abstraction over Product Arguments}
\label{sec:product-abstraction}

Although $\QuLam$ does not permit unrestricted abstraction over a surrounding environment, it permits several assumptions to be combined into a single product argument.

\begin{proposition}[Abstraction over a pair]
\label{prop:tuple-abstraction}
Suppose $x:\alpha,y:\beta\vdash M:\gamma$, where $x$ and $y$ are distinct. If $z$ is fresh, then the closed term $\lambda z.\,M[x:=\pi_1 z,y:=\pi_2 z]$ has type $(\alpha\times\beta)\rightarrow\gamma$. Here the substitutions are performed simultaneously and are assumed to avoid variable capture.
\end{proposition}

\begin{proof}
By Lemma~\ref{lem:substitution}, we obtain a derivation of $\{z:\alpha\times\beta\}\vdash M[x:=\pi_1z,y:=\pi_2z]:\gamma$ by simultaneously replacing $x:\alpha$ and $y:\beta$ with $\pi_1z:\alpha$ and $\pi_2z:\beta$, respectively. We then obtain the indicated closed term by the abstraction rule.
\end{proof}

This construction can be extended to any finite list of two or more assumptions by using right-associated products. In the case of a single assumption, we use the ordinary restricted abstraction rule.

Next, we consider the relation corresponding semantically to Proposition~\ref{prop:tuple-abstraction}. Let $a$, $b$, and $c$ be the interpretations of $\alpha$, $\beta$, and $\gamma$, respectively. From Lemma~\ref{lem:sasaki-adjunction} and $\varphi_d(1)=d$, it follows that $d\leq c$ and $d\rightarrow_S c=1$ are equivalent for every element $d$. Putting $d=a\wedge b$, we see that $a\wedge b\leq c$ and $(a\wedge b)\rightarrow_S c=1$ are equivalent.

We can also perform uncurrying by using product arguments. Indeed, if $M$ is a closed term of type $\alpha\rightarrow(\beta\rightarrow\gamma)$, then, for a variable $z$ of type $\alpha\times\beta$, $\pi_1z$ and $\pi_2z$ have types $\alpha$ and $\beta$, respectively. Therefore, $(M(\pi_1z))(\pi_2z)$ has type $\gamma$. Since $M$ is closed, the only free variable of this term is $z$. Hence, by the restricted abstraction rule, we obtain the closed term $U_M:=\lambda z.\,(M(\pi_1z))(\pi_2z)$ of type $(\alpha\times\beta)\rightarrow\gamma$.

\subsection{Failure of Ordinary Currying}
\label{sec:currying-failure}

The converse of uncurrying, namely ordinary currying, is not available in general. In particular, $\lambda$-terms such as $\lambda x.\lambda y.x$ and $\lambda x.\lambda y.\langle x,y\rangle$ cannot be typed under the abstraction rule of $\QuLam$, since $x$ occurs free in their inner abstractions.

The semantic fact corresponding to this failure is as follows. Let $\mathrm{MO}_2$ be the orthomodular lattice with six elements, consisting of $0$, $1$, and four pairwise incomparable elements $a$, $a^\perp$, $b$, and $b^\perp$. The pairs $\{a,a^\perp\}$ and $\{b,b^\perp\}$ belong to different Boolean subortholattices with four elements. If we interpret the type variables $\alpha$ and $\beta$ by $a$ and $b$, respectively, then, since $a\wedge b=0$ and $a\nleq b^\perp$, we have $b\rightarrow_S a=b^\perp\vee(b\wedge a)=b^\perp$ and $a\nleq b\rightarrow_S a$. Therefore, $\alpha\rightarrow(\beta\rightarrow\alpha)$ is not valid in all orthomodular lattices. By Theorem~\ref{thm:soundness}, no closed QuLam term has this type. Moreover, since $a\wedge b\leq a$ but $a\nleq b\rightarrow_S a$, implication does not furnish an exponential right adjoint to the product.\footnote{Malinowski~\cite{Mali1990} established a general obstruction to an ordinary deduction theorem for consequence relations determined by classes of orthomodular lattices.}

\subsection{Sasaki Projection}
\label{sec:sasaki-projection}

A different form of deduction is available through the Sasaki projection. In what follows, we formulate this deduction principle. For types $\alpha$ and $\gamma$, let $\varphi_\alpha(\gamma):=\alpha\times\neg(\alpha\times\neg\gamma)$. Under the interpretation introduced in Definition~\ref{def:interpretation}, fix a valuation in an orthomodular lattice, and write $a=\llbracket\alpha\rrbracket$ and $c=\llbracket\gamma\rrbracket$. Then $\llbracket\varphi_\alpha(\gamma)\rrbracket=a\wedge(a^\perp\vee c)$ is the Sasaki projection $\varphi_a(c)$ of $c$ onto $a$.

\begin{proposition}[Deduction using the Sasaki projection]
\label{prop:sasaki-projected-deduction}
Suppose $x:\varphi_\alpha(\gamma)\vdash M:\beta$. Then there exists a term of type $\alpha\rightarrow\beta$ under the single assumption $y:\gamma$. More precisely, we choose a fresh variable $z$ and put $K_M:=\lambda z.\,(\zeta(\tau(\lambda x.\langle\pi_1x,M\rangle,\pi_2z)))(\pi_1z)$. Since $M$ has no free variables other than $x$, $K_M$ is a closed term of type $(\alpha\times\neg(\alpha\times\beta))\rightarrow\neg\gamma$, and $y:\gamma\vdash\zeta(\tau(K_M,\theta(y))):\alpha\rightarrow\beta$ holds.
\end{proposition}

\begin{proof}
Under the assumption $x:\varphi_\alpha(\gamma)$, the pair $\langle\pi_1x,M\rangle$ has type $\alpha\times\beta$. Therefore, $F_M:=\lambda x.\langle\pi_1x,M\rangle$ is a closed term of type $\varphi_\alpha(\gamma)\rightarrow(\alpha\times\beta)$. Now assume $z:\alpha\times\neg(\alpha\times\beta)$. The second projection $\pi_2z$ has type $\neg(\alpha\times\beta)$, and therefore $\tau(F_M,\pi_2z)$ has type $\neg\varphi_\alpha(\gamma)=\neg(\alpha\times\neg(\alpha\times\neg\gamma))$. Applying $\zeta$, we obtain a term of type $\alpha\rightarrow\neg\gamma$ under $z$. When we apply this term to $\pi_1z:\alpha$, the resulting term has type $\neg\gamma$. Therefore, by abstraction over $z$, we can construct $K_M:(\alpha\times\neg(\alpha\times\beta))\rightarrow\neg\gamma$ with no free variables. Moreover, by the double-negation introduction rule, $y:\gamma\vdash\theta(y):\neg\neg\gamma$ holds. By modus tollens, $\tau(K_M,\theta(y))$ has type $\neg(\alpha\times\neg(\alpha\times\beta))$ under $y:\gamma$. Furthermore, applying $\zeta$, we obtain a term of type $\alpha\rightarrow\beta$ under the same assumption.
\end{proof}

We put $b=\llbracket\beta\rrbracket$. By soundness, we obtain $\varphi_a(c)\leq b$ from the hypothesis of Proposition~\ref{prop:sasaki-projected-deduction} and $c\leq a\rightarrow_S b$ from its conclusion. These inequalities correspond to the forward direction of Lemma~\ref{lem:sasaki-adjunction}.\footnote{The hypothesis of Proposition~\ref{prop:sasaki-projected-deduction} has the same semantic condition as Kornell's ordered sequent $\gamma,\alpha\vdash\beta$. In SOM~\cite{Korn2024}, this sequent is interpreted as $c\mathbin{\&}a\leq b$, where $c\mathbin{\&}a=a\wedge(a^\perp\vee c)$.} If $a$ and $c$ commute in an orthomodular lattice, the projection $a\wedge(a^\perp\vee c)$ equals $a\wedge c$. Therefore, for commuting $a$ and $c$, this deduction principle has a premise interpreted by their meet.

\section{Conclusion}\label{sec:conclusion}

In this paper, we have refined the natural deduction system NQ for quantum logic and its corresponding typed lambda calculus QuLam, presented in our previous paper, and introduced rules for the Sasaki hook and the corresponding term constructors $\zeta$ and $\xi$. Furthermore, for these extended systems, we have established the soundness and completeness of NQ, the Curry--Howard correspondence between the two systems, the strong normalization of QuLam, and the soundness and completeness of typing judgments.

In QuLam, the body of a $\lambda$-abstraction must be typed under a single assumption, and hence ordinary currying is not available in general. However, we can still perform $\lambda$-abstraction by combining multiple assumptions into a single product argument, uncurry closed terms, and realize deduction based on the Sasaki projection as a term.

\section*{Acknowledgments}

This work was supported by JSPS KAKENHI Grant Number JP24K03372.


\begin{thebibliography}{99}

\bibitem{Alte2005}
T. Altenkirch and J. Grattage,
A functional quantum programming language,
in \emph{20th Annual IEEE Symposium on Logic in Computer Science (LICS'05)},
IEEE, 2005, 249--258.
\href{https://doi.org/10.1109/LICS.2005.1}{\nolinkurl{doi:10.1109/LICS.2005.1}}.

\bibitem{Ario2003}
Z. M. Ariola and H. Herbelin,
Minimal classical logic and control operators,
in \emph{Automata, Languages and Programming},
Lecture Notes in Computer Science, Vol. 2719,
Springer, Berlin and Heidelberg, 2003, 871--885.
\href{https://doi.org/10.1007/3-540-45061-0_68}{\nolinkurl{doi:10.1007/3-540-45061-0_68}}.

\bibitem{Birk1936}
G. Birkhoff and J. von Neumann,
The Logic of Quantum Mechanics,
\emph{Annals of Mathematics} 37(4) (1936), 823--843.
\href{https://doi.org/10.2307/1968621}{\nolinkurl{doi:10.2307/1968621}}.

\bibitem{Coec2004}
B. Coecke and S. Smets,
The Sasaki hook is not a [static] implicative connective but induces a backward [in time] dynamic one that assigns causes,
\emph{International Journal of Theoretical Physics} 43(7--8) (2004), 1705--1736.
\href{https://doi.org/10.1023/B:IJTP.0000048815.92983.6e}{\nolinkurl{doi:10.1023/B:IJTP.0000048815.92983.6e}}.

\bibitem{Gira1989}
J.-Y. Girard,
\emph{Proofs and Types}, translated and with appendices by P. Taylor and Y. Lafont,
Cambridge Tracts in Theoretical Computer Science, Vol. 7,
Cambridge University Press, Cambridge, 1989.

\bibitem{Grif1990}
T. G. Griffin,
A formulae-as-type notion of control,
in \emph{Proceedings of the 17th ACM SIGPLAN-SIGACT Symposium on Principles of Programming Languages},
ACM Press, New York, 1990, 47--58.
\href{https://doi.org/10.1145/96709.96714}{\nolinkurl{doi:10.1145/96709.96714}}.

\bibitem{Hard1979}
G. M. Hardegree,
The conditional in abstract and concrete quantum logic,
in C. A. Hooker, ed., \emph{The Logico-Algebraic Approach to Quantum Mechanics, Volume II: Contemporary Consolidation},
D. Reidel, Dordrecht, 1979, 49--108.

\bibitem{Hard1981}
G. M. Hardegree,
Material implication in orthomodular (and Boolean) lattices,
\emph{Notre Dame Journal of Formal Logic} 22(2) (1981), 163--182.
\href{https://doi.org/10.1305/ndjfl/1093883401}{\nolinkurl{doi:10.1305/ndjfl/1093883401}}.

\bibitem{Howa1980}
W. A. Howard,
The formulae-as-types notion of construction,
in J. R. Hindley and J. P. Seldin, eds., \emph{To H.B. Curry: Essays on Combinatory Logic, Lambda Calculus and Formalism},
Academic Press, London, 1980, 479--490.

\bibitem{Kalm1983}
G. Kalmbach,
\emph{Orthomodular Lattices},
Academic Press, London and New York, 1983.

\bibitem{Korn2024}
A. Kornell,
A natural deduction system for orthomodular logic,
\emph{The Review of Symbolic Logic} 17(3) (2024), 910--949.
\href{https://doi.org/10.1017/S1755020323000229}{\nolinkurl{doi:10.1017/S1755020323000229}}.

\bibitem{Mali1990}
J. Malinowski,
The deduction theorem for quantum logic---some negative results,
\emph{The Journal of Symbolic Logic} 55(2) (1990), 615--625.
\href{https://doi.org/10.2307/2274651}{\nolinkurl{doi:10.2307/2274651}}.

\bibitem{Nish1980}
H. Nishimura,
Sequential method in quantum logic,
\emph{The Journal of Symbolic Logic} 45(2) (1980), 339--352.
\href{https://doi.org/10.2307/2273194}{\nolinkurl{doi:10.2307/2273194}}.

\bibitem{Nish2009}
H. Nishimura,
Gentzen methods in quantum logic,
in K. Engesser, D. M. Gabbay, and D. Lehmann, eds.,
\emph{Handbook of Quantum Logic and Quantum Structures: Quantum Logic},
Elsevier, Amsterdam, 2009, 227--260.
\href{https://doi.org/10.1016/B978-0-444-52869-8.50009-8}{\nolinkurl{doi:10.1016/B978-0-444-52869-8.50009-8}}.

\bibitem{Pari1992}
M. Parigot,
$\lambda\mu$-calculus: an algorithmic interpretation of classical natural deduction,
in A. Voronkov, ed., \emph{Logic Programming and Automated Reasoning},
Lecture Notes in Computer Science, Vol. 624,
Springer, Berlin and Heidelberg, 1992, 190--201.
\href{https://doi.org/10.1007/BFb0013061}{\nolinkurl{doi:10.1007/BFb0013061}}.

\bibitem{Seli2005}
P. Selinger and B. Valiron,
A lambda calculus for quantum computation with classical control,
in \emph{Typed Lambda Calculi and Applications},
Lecture Notes in Computer Science, Vol. 3461,
Springer, Berlin and Heidelberg, 2005, 354--368.
\href{https://doi.org/10.1007/11417170_26}{\nolinkurl{doi:10.1007/11417170_26}}.

\bibitem{Sore2006}
M. H. S{\o}rensen and P. Urzyczyn,
\emph{Lectures on the Curry--Howard Isomorphism},
Studies in Logic and the Foundations of Mathematics, Vol. 149,
Elsevier, Amsterdam, 2006.

\bibitem{Toku2022}
K. Tokuo,
Natural deduction for quantum logic,
\emph{Logica Universalis} 16(3) (2022), 469--497.
\href{https://doi.org/10.1007/s11787-022-00307-7}{\nolinkurl{doi:10.1007/s11787-022-00307-7}}.

\bibitem{Vant2004}
A. Van Tonder,
A lambda calculus for quantum computation,
\emph{SIAM Journal on Computing} 33(5) (2004), 1109--1135.
\href{https://doi.org/10.1137/S0097539703432165}{\nolinkurl{doi:10.1137/S0097539703432165}}.

\end{thebibliography}
\end{document}